%% file: wsdm_20.tex
\documentclass[sigconf,natbib,acmthm]{acmart} 
\newcommand{\compilefullversion}{false} 
\newcommand{\compilehidecomments}{false}

\definecolor{red}{HTML}{E51400} 
\definecolor{blue}{HTML}{0050EF} 
\definecolor{green}{HTML}{008A00} 
\definecolor{purple}{HTML}{AA00FF} 
\definecolor{orange}{HTML}{FF7F00}
\definecolor{gray}{HTML}{848482} 

\usepackage{booktabs} 

\usepackage{hyperref}
\hypersetup{colorlinks, linkcolor=blue, citecolor=green, anchorcolor=blue, urlcolor=black}  
\usepackage{nicefrac}       
\usepackage{microtype}      
\usepackage{ifthen}
\usepackage{epstopdf}   
\usepackage[ruled,vlined,linesnumbered]{algorithm2e}
\usepackage{multirow}
\usepackage{thm-restate}
\setcitestyle{numbers,sort&compress}
\usepackage{sistyle}
\usepackage{graphicx}
\usepackage{subfig}

\DeclareMathOperator*{\argmax}{argmax}
\DeclareMathOperator*{\argmin}{argmin}

\newcommand{\E}{\mathbb{E}}
\newcommand{\I}{\mathbb{I}}
\newcommand{\R}{\mathbb{R}}

\newcommand{\cP}{\mathcal{P}}
\newcommand{\cR}{\mathcal{R}}

\newcommand{\cW}{\mathcal{W}}

\newcommand{\rroot}{{\rm root}}

\newcommand{\covered}{{\it covered}}
\newcommand{\delay}{{\it delay}}
\newcommand{\shadow}[1]{\hat{#1}}

\newcommand{\OPT}{{\rm OPT}}

\newcommand{\topk}{{\it topk}}

\newcommand{\PRR}{\mbox{\sf P-RR}}
\newcommand{\NodeSelection}{{\sf NodeSelection}}

\newcommand{\BIM}{\mbox{BIM}}
\newcommand{\PIM}{\mbox{PIM}}
\newcommand{\BPIM}{\mbox{BPIM}}

\newcommand{\IMM}{\mbox{\sf IMM}}
\newcommand{\ASVRR}{\mbox{\sf ASV-RR}}
\newcommand{\IMMBIM}{\mbox{\sf IMM-BIM}}
\newcommand{\IMMPIM}{\mbox{\sf IMM-PIM}}
\newcommand{\IMMBPIM}{\mbox{\sf IMM-BPIM}}

\newcommand{\alg}[1]{\textnormal{\textsf{#1}}}

\newcommand{\newalg}[1]{{#1}}
\SetNlSty{bfseries}{\color{black}}{}

\SetCommentSty{mycommfont}

\ifthenelse{\equal{\compilefullversion}{false}}{%
	\newcommand{\OnlyInFull}[1]{}
	\newcommand{\OnlyInShort}[1]{#1}
}{%
	\newcommand{\OnlyInFull}[1]{#1}%
	\newcommand{\OnlyInShort}[1]{}%
}%

\ifthenelse{ \equal{\compilehidecomments}{true} }{%
    \newcommand{\wei}[1]{}
    \newcommand{\weiran}[1]{}
    \newcommand{\lichao}[1]{}
    \newcommand{\backup}[1]{}
}{
    \newcommand{\wei}[1]{{\color{blue!50!black}  [\text{Wei:} #1]}}
    \newcommand{\weiran}[1]{{\color{red!70!black} [\text{Weiran:} #1]}}
    \newcommand{\lichao}[1]{{\color{green!90!black} [\text{Lichao:} #1]}}
    \newcommand{\backup}[1]{{\color{green!50!black} #1}}
}

\begin{document}

\title{Influence Maximization with Spontaneous User Adoption}


\author{Lichao Sun}
\email{james.lichao.sun@gmail.com}
\affiliation{%
	\institution{University of Illinois at Chicago}
	\city{Chicago}
	\state{IL}
	\postcode{60607}
}

\author{Albert Chen}
\email{abchen@linkedin.com,}
\affiliation{%
	\institution{LinkedIn Corporation}
	\city{Sunnyvale}
	\state{CA}
	\postcode{94085}
}

\author{Philip S. Yu}
\email{psyu@uic.edu}
\affiliation{%
	\institution{University of Illinois at Chicago}
	\city{Chicago}
	\state{IL}
	\postcode{60607}
}

\author{Wei Chen}
\authornote{Corresponding author}
\email{weic@microsoft.com}
\affiliation{%
	\institution{Microsoft Research}
	\city{Beijing}
	\state{China}
	\postcode{100080}
}

%
%
%
%
%
%


\input{sec_abstract}

%
%



\maketitle

\sloppy

\input{sec_intro}

\input{sec_related}

\input{sec_model}
\input{sec_properties}

\input{sec_imm}
\input{sec_exp}

\bibliographystyle{ACM-Reference-Format}
\bibliography{bibdatabase} 
\clearpage

\input{appendix}

\end{document}

%% file: sec_abstract.tex
\begin{abstract}

We incorporate the realistic scenario of spontaneous user adoption into influence propagation (also refer to as self-activation) and propose
	the self-activation independent cascade (SAIC) model: nodes may be self activated besides
	being selected as seeds, and influence propagates from both selected seeds and
	self activated nodes.
Self activation occurs in many real world situations; for example, people naturally share product recommendations with their friends, even without marketing intervention.
Under the SAIC model, we study three influence maximization problems: 
	(a) boosted influence maximization (BIM) aims to maximize the total influence spread
		from both self-activated nodes and $k$ selected seeds;
	(b) preemptive influence maximization (PIM) aims to find $k$ nodes that, if self-activated, can reach the
		most number of nodes before other self-activated nodes; and
	(c) boosted preemptive influence maximization (BPIM) aims to select $k$ seed that are guaranteed
		to be activated and can reach the most number of nodes before other self-activated nodes.
We propose scalable algorithms for all three problems and prove that they achieve 
	$1-1/e-\varepsilon$ approximation for BIM and BPIM and $1-\varepsilon$ for PIM, for
	any $\varepsilon > 0$.
Through extensive tests on real-world graphs, we demonstrate that our algorithms outperform the baseline algorithms
	significantly for the PIM problem in solution quality, and also outperform the baselines for BIM and BPIM
	when self-activation behaviors are nonuniform across nodes. 

\end{abstract}

%% file: sec_intro.tex
\vspace{-4pt}
\section{Introduction}
\vspace{-4pt}
Influence maximization is the task of finding a small set of seed nodes to generate the largest possible influence spread in
	a social network~\cite{kempe03}.
It models the important viral marketing applications in social networks,
	and many aspects of influence maximization have been extensively studied in the research literature.
In most studies, influence propagation starts from a set of seed nodes, which are selected before the propagation starts.
Propagation starts from all seed nodes together at the same time and proceeds in either discrete or continuous time to 
	reach other nodes in the network, and the objective to maximize is typically the {\em influence spread}, defined as 
	the expected number of nodes activated through the stochastic diffusion process.
	
In practice, however, when a marketing campaign starts, users' reactions to the campaign are not synchronized at the same time. 
Some users react to the campaign immediately, while others may react after a significant delay. 
Moreover, propagation may not only start from seed users that the campaign originally selected.
It is possible that other users may spontaneously react to the campaign and also propagate the information and influence in the campaign.
We call this phenomenon {\em self activation},
	which is in contrast to the seed activation by the external force.
While seed activation by the external force 
	typically requires a marketing budget to be successful, self activations are spontaneous and do not require
	a budget.
Self activation may also lead to other interesting objectives one may want to optimize, as we will discuss shortly.
To give some concrete examples of self-activation without marketing influence,
people may naturally share consumer product recommendations with their friends and businesses may attract clients through organic referrals.
A product or business marketing team interested in influence maximization should consider this natural activity.
Therefore, self activation is a realistic phenomenon in viral marketing, but it has not been well addressed in the influence
	maximization literature.
In this paper, we incorporate self activation into the influence propagation model and 
	provide a systematic study on the impact of self activation to the influence maximization task.
	
We first incorporate self activation with the classical independent cascade (IC) model to propose the self-activation
	independent cascade (SAIC) model of influence propagation.
In the SAIC model, social network is modeled as a directed graph, and each node $u$ has a {\em self-activation probability} $q(u)$
	indicating the probability of $u$ being self activated by the campaign.
If activated, node $u$ also has a random {\em self-activation delay} $\delta(u)$ sampled from distribution $\Delta(u)$, such that 
	$u$ is self-activated at time $\delta(u)$ after the campaign starts at time $0$.
A seed node $v$ selected explicitly by the campaign will be deterministically activated at time $\delta(v)$, equivalently as 
	saying that its self-activation probability $q(v)$ is boosted to $1$.
Propagation from seed nodes and self-activated nodes follows the classical IC model: if a node $u$ is activated, then it has
	one chance to activate each of its out-going neighbor $v$ with success probability $p(u,v)$.
We further extend the classical IC model by allowing real-time delay on the edges: if $u$ would successfully activate $v$, then
	this activation would occur after the random {\em propagation delay} of $d(u,v)$ sampled from a distribution $D(u,v)$, from the time $u$
	is activated.
Thus, overall the propagation starts from the seed nodes and self-activated nodes, and a node is activated either because it is
	a seed, or because it is self-activated, or because it is activated by a neighbor, and the activation time is the earliest time
	when one of the above activation happens.
Once a node is activated, it stays as activated.
The reason we allow real time delays is to make it more realistic when we study the new objective functions discussed below.

With the incorporation of self activation into the SAIC model, we are able to study several different influence maximization objectives.
The first objective is closer to the classical influence maximization, where we aim at select $k$ seed nodes 
	to maximize the total influence spread after boosting the self-activation probabilities of seed nodes to $1$.
We refer to this objective as {\em boosted influence spread}, denoted $\sigma^B(S)$ for seed set $S$, and the corresponding 
	influence maximization problem as {\em boosted influence maximization (BIM)}.
Note that the total boosted influence spread $\sigma^B(S)$ counts both nodes activated by seeds and nodes activated by 
	self-activated nodes.

Besides $\sigma^B(S)$, self activation further allows us to study some interesting new objectives.
Conceptually, in the SAIC model, we can view a node as an {\em organic influencer} if
	the node is frequently and easily self-activated and its influence often reaches many
	other nodes first before other self-activated nodes.
To model this, we precisely define the {\em preemptive influence spread} 
	$\rho(A)$ of node set $A$ as the expected number of nodes that some node $u\in A$ 
	reaches first, if $u$ is self-activated, 
	before other self-activated nodes.
Then the {\em preemptive influence maximization (PIM)} problem is to identify the set of $k$
	nodes that has the largest preemptive influence spread, which models the task of
	identifying top organic influencers in a network.
Furthermore, we define {\em boosted preemptive influence spread} $\rho^B(S)$ of 
	a seed set $S$
	as the preemptive influence spread of set $S$ after we boost the self-activation
	probability of every node in $S$ to $1$.
Then the {\em boosted preemptive influence maximization (BPIM)} problem is to find
	$k$ seed nodes with the maximum boosted preemptive influence spread $\rho^B(S)$.
BPIM corresponds to the viral marketing campaign that focuses on the reach of the campaign
	from the selected seed nodes rather than self-activated nodes, because for example the
	seed nodes carry high-quality and more effective campaign messages.

	
For the above objectives, we first study their properties and show that $\sigma^B$ and $\rho^B$ are monotone and submodular, while
	$\rho$ is additive.
Then, based on these properties and the reverse influence sample (RIS) approach~\cite{BorgsBrautbarChayesLucier}, we design scalable approximation algorithms for 
	the three problems BIM, BPIM, and PIM.
Even though our algorithms are patterned from the existing algorithm, the new objectives studied here requires nontrivial adaptation
	of the algorithm.
Especially for BPIM and PIM, we need to redesign the reverse simulation procedure to generate what we call {\em preemptive
	reverse reachable (P-RR) sets}, which are more sophisticated than the standard reverse reachable (RR) sets~\cite{BorgsBrautbarChayesLucier,tang14,tang15}.
We prove that our algorithms solve BIM and BPIM with approximation ratio $1-1/e -\varepsilon$ for any $\varepsilon > 0$, and
	solve PIM with approximation ratio $1-\varepsilon$, and all algorithms can run in time near linear to the graph size.
	
Finally, we conduct extensive experiments on real-world datasets and compare our algorithms with related baselines solving 
	classical influence maximization or influence-based network centrality.
We demonstrate that for the PIM problem, our {\IMMPIM} algorithm 
	significantly outperforms the baselines on the achieved preemptive
	influence spread in all test cases, showing that utilizing the knowledge of self activation is important in finding organic influential nodes in a network.
For the BIM and BPIM problems, our algorithms have minor improvements in influence spread in some cases where the self-activation behaviors of the nodes
	are non-uniform, which algorithms such as {\IMM} that are oblivious to 
	self-activation may be used for these problems, 
	but it may still be beneficial to use our algorithms designed for the self-activation scenarios.
All our algorithms can scale to large networks with hundred thousands of nodes and edges.

To summarize, we have the following contributions:
	(a) we incorporate the realistic self-activation scenario into influence propagation and study three influence maximization
		problems PIM, BPIM, and BIM due to this incorporation;
	(b) we design scalable algorithms for all three problems with theoretical approximation guarantee; and
	(c) we demonstrate through experiments that our algorithms provide significantly better results for PIM, and also outperform
		other algorithms in non-uniform self-activation scenarios.

Due to the space constraint, some experiments, some pseudocode and all proofs are moved to the extended version \cite{sun2019self}.

%% file: sec_related.tex
\vspace{-5pt}
\subsection{Related Work}

Domingos and Richardson are the first to study influence maximization~\cite{domingos01,richardson02}, 
	but \citet{kempe03} are the first to formulate the problem as a discrete
	optimization problem, describe the independent cascade (IC), linear threshold and other models, 
	and propose to use submodularity and greedy algorithm to solve influence maximization.
Since then influence maximization has been extensively studied in various fronts, 
	including its scalability~\cite{ChenWY09,ChenYZ10,WCW12,JungHC12,BorgsBrautbarChayesLucier,tang14,tang15,NguyenTD16,sun2018multi}, 
	robust influence maximization~\cite{ChenLTZZ16,HeKempe16}, 
	competitive and complementary
	influence maximization~\cite{kempe07,BAA11,chen2011influence,HeSCJ12,lu2015competition}, etc.

\citet{BorgsBrautbarChayesLucier} propose the novel reverse influence sampling (RIS) approach that guarantees both the approximation ratio
	and near-linear running time, and it has been improved in a series of studies~\cite{BorgsBrautbarChayesLucier,tang14,tang15,NguyenTD16}.
In this paper, we adapt the {\IMM} algorithm in~\cite{tang15} mainly because of its clarity, and other algorithms such as {\sf D-SSA}
	of~\cite{NguyenTD16} can
	be plugged in too.

Our PIM problem has connection with the Shapley centrality proposed in \cite{chen2017interplay}.
In particular, in the special case when all nodes have self-activation probability $1$, uniform self-activation delay distribution, 
	and propagation delays can be ignored, the preemptive influence spread of a node coincides with its Shapley centrality.
Thus, the general preemptive influence spread studied in this paper is more realistic than the Shapley centrality, and we compare
	against the Shapley centrality algorithm in our experiments and show that our algorithm achieves much better PIM result.

%% file: sec_model.tex
\vspace{-5pt}
\section{Model and Problem Definition} \label{sec:model}

\subsection{Self-Activation Propagation Model}


A social network is modeled as a directed
graph $G = (V, E)$,
where $V$ is a finite set of vertices or nodes, and
$E \subseteq V \times V$ is the set of directed edges connecting pairs of nodes.
Let $n=|V|$ and $m=|E|$.
In this paper, we study the influence propagation model
where every node has a chance to be self-activated, even without being selected as a seed.
In this model, at time $0$ we assume a marketing campaign is started, and then
	every node may be self-activated by this campaign, and this activation may occur
	after a random delay from the beginning of the campaign.
Technically node self activation is governed by the following set of parameters.


\begin{definition} \label{def:SAProbDelay}
	(Self-Activation Probability and Delay).
	In a social network, every node $u \in V$ can be self-activated as a seed by the campaign with
	{\em self-activation probability} $q(u) \in [0,1]$.
	If $u$ is self-activated, then it is activated after a random delay $\delta(u)\in [0,+\infty)$ drawn from a 
	{\em self-activation delay} distribution $\Delta(u)$.
\end{definition}

We combine self activation and independent cascade model~\cite{kempe03}
and further add real-time propagation delays on edges to obtain the self-activation independent cascade (SAIC) model.
In the SAIC model, every node $u \in V$ is associated with self-activation probability and delay
	as defined in~Definition~\ref{def:SAProbDelay}.
Meanwhile, every edge $(u,v) \in E$ is associated with two parameters:
1) a {\em propagation probability} $p(u,v) \in (0, 1]$ ($p(u,v)=0$ if and only if $(u,v) \notin E$);
2) a {\em propagation delay} distribution $D(u,v)$ with range $[0,+\infty)$.
The SAIC model proceeds following the rules below:
\begin{enumerate}
	\item For each node $u \in V$, it is self-activated with probability $p(u)$,
	and if so it is activated at time $\delta(u)$ drawn from $\Delta(u)$ (denoted as $\delta(u)\sim \Delta(u)$),
	unless it has been previously activated by other nodes before $\delta(u)$.
	\item For any node $u \in V$ activated at time $t$ (self-activated or neighbor-activated),
	it tries once to activate each of its out-going neighbors $v$ 
	with propagation probability $p(u,v)$.
	If the activation is successful,
	the propagation delay $d(u,v)$ is sampled from $D(u,v)$
	and $v$ would be activated at time $t + d(u,v)$,
	unless $v$ has been activated before this time.
	\item A node $v$ is activated (both self-activated or neighbor-activated) at earliest time $t$
	when any activation of $v$ would happen, and $v$ stays active afterwards.
\end{enumerate}
The above description of the SAIC model only considers the propagation starting from the
	self-activated nodes.
We could further include externally selected {\em seed nodes} in the model as follows.
Let $S$ be the set of (externally selected) seed nodes.
Then besides the three rules above, we have one more rule:
\begin{enumerate}
	\item[(4)] For each seed $u\in S$, its self-activation probability is boosted to $1$, that is, 
	it is for sure activated, and it is activated at time $\delta(u)\sim \Delta(u)$,
	unless it has been activated by its neighbors before that time.
\end{enumerate}

Note that when the self-activation probabilities of all nodes are $0$, the self-activation delays
	for all nodes take deterministic value $0$, 
	and propagation delays of all nodes take 
	deterministic value $1$, 
	SAIC model falls back to the classical independent cascade (IC) model.
We add self activation to model the realistic situation that some users may spontaneously
	help propagating the marketing campaign.
We introduce delay parameters because in reality people take actions asynchronously, and when we
	consider preemptive influence spread (to be defined shortly), delay factor does matter on who
	could claim credits on the activation of each node.

The SAIC model can be equivalently described as propagations in a {\em possible world model}.
A possible world $W$ contains all randomness in a SAIC propagation.
In particular, $W$ is a tuple $(A_W,\delta_W,L_W,d_W)$, where $A_W$ is the random set of all self-activated
	nodes governed by the self-activation probability $q$, $\delta_W$ is the vector of self-activation
	delays sampled from $\Delta$, $L_W$ is the set of live edges governed by the propagation probability
	$p$ (i.e. each edge $(u,v)\in E$ is live with probability $p(u,v)$), and $d_W$ is the vector
	of propagation delays sampled from $D$.
We use $\cW(q,\Delta,p,D)$ to denote the probability space of all possible worlds, determined by
	parameters $q,\Delta,p,D$.
In a fixed possible world $W$, a live path $P$ is a path consisting of only live edges in $L_W$.
The propagation delay $d_W(P)$ of a live path is the sum of propagation delays on all live edges.
For a $u\in A_W$ and a live path $P$ from $u$ to $v$ in $W$, 
	let $T_W(P) = \delta_W(u) + d_W(P)$ be the {\em total delay} of live path $P$. 
We also use $T_W(u,v)$ to denote the minimum total delay among all live paths from $u$ to $v$.
Propagation starts from nodes in $A_W$ and follows the direction of all live edges and incurs delay on
	the edges, and a node $v$ is activated at a time $t$ if the minimum total delay of all 
	live paths from any node in $A$ to $v$ is $t$.
If we have externally selected seed set $S$, then the propagation starts from $S\cup A$ instead of $A$,
	and all other aspects remain the same.
It is easy to see that the above described possible world model is just a different way of stating
	the SAIC model, when we determine all randomness before the propagation starts.

\vspace{-5pt}
\subsection{Self-Activation Influence Maximization}


In the SAIC model, we consider several optimization tasks, each corresponding to a different
	objective function.
We start with an objective function that is close to the influence spread objective function in
	the classical IC model.
For a possible world $W$ and a seed set $S\subseteq V$, let $\Phi^B_W(S)$ denote the 
	set of nodes activated in the possible world $W$ 
	when $S$ is the seed set (nodes in $S$ have their self-activation probabilities boosted to $1$).
We call $\sigma^B(S) = \E_{W\sim \cW(q,\Delta,p,D)}[|\Phi^B_W(S)|]$ the {\em boosted influence spread}
	of seed set $S$ in the SAIC model.
Note that $\sigma^B(\emptyset)$ may be greater than $0$ since some nodes may be self-activated,
	and thus we use the word ``boosted'' to refer to the final influence spread due to the
	boosting of self-activation probabilities of seed nodes to $1$.
It is easy to see that the delay distributions $\Delta$ and $D$ do not affect the final set of
activated nodes, and thus $\sigma^B(S)$ doesn't depend on $\Delta$ and $D$.
Moreover, when $q \equiv 0$, $\sigma^B(S)$ becomes the influence spread in the classical IC model.
The first optimization task is to maximize $\sigma^B(S)$:

\begin{definition}[Boosted Influence Maximization]
\label{def:boostedInfMax}
{\em Boosted influence maximization (\BIM)}  
	is the optimization task with the directed graph $G=(V,E)$,
	the self-activation probabilities $q$, the propagation probabilities $p$, and a budget $k$ as the input,
	and the goal is to find an optimal seed set $S^*$ having at most
	$k$ nodes, such that the {\em boosted influence spread} of  $S^*$ is maximized, i.e.,
	$S^* = \argmax_{|S| \leq k} \sigma^B(S)$.
\end{definition}

	

The boosted influence spread and boosted influence maximization is close to the classical 
	influence spread and influence maximization concepts.
We introduce them as a stepping stone for the new concept of preemptive influence spread and
	preemptive influence maximization.
	
In the SAIC model, a natural metric measuring the influence ability of a set of nodes $S$ is
	the number of nodes that are actually activated due to the propagation from nodes in $S$, 
	not by other sources.
We define this metric formally as follows.
First, we assume that all delay distributions in $\Delta$ and $D$ are continuous functions and thus
	there is no probability mass at any given value.
Thus, it is safe to assume that in any possible world $W$, the total delay of any path would be different,
	since the worlds with paths having same delays have probability measure $0$.
Given a possible world $W = (A_W, \delta_W, L_W, d_W)$,  
	let $\cP_W(u,v)$ denotes a set of all live paths in $W$ starting from node $u$ and ending at $v$.
For a set of nodes $A$, we use $\Gamma_W(A)$ to denote the set of nodes $v$
	that have minimum total delays from some node in $A$ to $v$, i.e.
	$\Gamma_W(A) = \{v \mid  \exists u \in A_W \cap A, \exists P \in \cP_W(u,v), 
		T_W(P) < +\infty, \forall u' \in A_W \setminus A, \forall P' \in \cP_W(u',v), 
		T_W(P) < T_W(P')\}$.
Set $\Gamma_W(A)$ contains all activated nodes in $W$ whose activation sources are some nodes in $A$.
In other words, $A$ could claim full credits for activating $\Gamma_W(A)$ in $W$.
We define $\rho(A) = \E_{W \sim\cW(q,\Delta,p,D)} [|\Gamma_W(A)|] $ the {\em preemptive 
	influence spread} of node set $A$.
Intuitively, preemptive influence spread $\rho(A)$ measures the contribution of node set $A$
	in propagating and activating nodes in the SAIC model, when there is no externally selected seed
	nodes.

When we use preemptive influence spread $\rho(A)$ as our objective function, we have the
	preemptive influence maximization problem.
	
\begin{definition}[Preemptive Influence Maximization]
\label{def:preemptiveInfMax}
{\em Preemptive influence maximization (\PIM)}  
is the optimization task with the directed graph $G=(V,E)$,
the self-activation probabilities $q$, 
the self-activation delay distribution $\Delta$, the propagation probabilities $p$, the
	propagation delay distribution $D$, and a budget $k$ as the input,
and the goal is to find an optimal set $A^*$ having at most
$k$ nodes, such that the {\em preemptive influence spread} of  $A^*$ is maximized, i.e.,
$A^* = \argmax_{|A| \leq k} \rho(A)$.
\end{definition}


Preemptive influence maximization defined above corresponds to the application where a company may want to identify top organic influencers in an online social network, to study the characteristics that make them influential. These users could be targeted to propagate company-specific information without changing their activation behavior.

We remark that, when the self-activation probabilities of all nodes are $1$, self-activation delays
	of all nodes follow the same distribution, 
	and propagation delays of all nodes take 
	deterministic value $0$ (i.e. propagations are instantaneous), 
	preemptive influence spread of each individual node
	coincides with the Shapley centrality defined in~\cite{chen2017interplay}.

In preemptive influence maximization, we do not have the action of selecting seeds and changing the behavior
	of seeds.
We can further incorporate seed selection with preemptive 
	influence maximization as follows.	
In a possible world $W = (A_W, \delta_W, L_W, d_W)$,  for a seed set $S$, 
	similar to $\Gamma_W(A)$ we define $\Gamma^B_W(S)$ as the set of nodes that are activated due to
	$S$, after nodes in $S$ are selected as seeds and their self-activation probabilities are boosted
	to $1$, that is,
	$\Gamma^B_W(S) = \{v \mid  \exists u \in S, \exists P \in \cP_W(u,v), 
T_W(P) < +\infty, \forall u' \in A_W \setminus S, \forall P' \in \cP_W(u',v), 
T_W(P) < T_W(P')\}$.
We define the {\em boosted preemptive influence spread} $\rho^B(S)$ of a seed set $S$ 
	as $\rho^B(S) = \E_{W \sim\cW(q,\Delta,p,D)} [|\Gamma^B_W(S)|] $.
We can now use $\rho^B(S)$ as our third objective function to define the third optimization task:

\begin{definition}[Boosted Preemptive Influence Maximization]
	\label{def:BpreemptiveInfMax}
	{\em Boosted preemptive influence maximization (\BPIM)}  
	is the optimization task with the directed graph $G=(V,E)$,
	the self-activation probabilities $q$, 
	the self-activation delay distribution $\Delta$, the propagation probabilities $p$, the
	propagation delay distribution $D$, and a budget $k$ as the input,
	and the goal is to find an optimal seed set $S^*$ having at most
	$k$ nodes, such that the {\em boosted preemptive influence spread} of  $S^*$ is maximized, i.e.,
	$S^* = \argmax_{|S| \leq k} \rho^B(S)$.
\end{definition}

BPIM defined above models the applications where the marketing campaign wants to engage in explicit
	incentive for a set of seed nodes (e.g. giving out free sample products) so that the seed nodes
	will be boosted to adopt the campaign and start propagating it.
The difference between BPIM and BIM (Definition~\ref{def:boostedInfMax}) is that BPIM only optimizes
	for the number of nodes first activated by the seed nodes, while BIM optimizes for the number
	of all activated nodes.
The difference between BPIM and PIM (Definition~\ref{def:preemptiveInfMax}) is that BPIM actively
	boosts the self-activation probabilities of seed nodes while PIM does not change node behaviors, and
	this is also the reason in PIM we avoid calling the set $A$ selected as a seed set.
Although all three problems look similar on surface, they are different and require
	separate algorithmic solutions.
Moreover, PIM is very different from BIM and BPIM in that their algorithmic solutions would have
	different approximation guarantees.
This is because the preemptive influence spread has different properties from the other 
	two objective functions, as we discuss in the next section.
	
Finally, we remark that our proposed SAIC model is rich enough to consider realistic self-activation
	scenarios and all the above optimization tasks, while we do not introduce further parameters to
	complicate situation.
For example, for a seed node $u\in S$, we could further consider shortening its self-activation
	delay or boosting its self-activation probability partially instead of $1$.
The added flexibility would not significantly change our algorithm design and analysis but only 
	complicate our presentation.
On the other hand, assuming seed selection would not shorten the self-activation delay is also reasonable,
	since when an online marketing campaign starts, a user in the network need to come online to be
	aware of this marketing campaign, and thus the initial delay from the time the campaign starts
	to the user coming online is not likely affected by the user being selected as a seed.
Finally, the added parameters, such as self-activation delay distributions, self-activation probabilities
	are likely to be empirically obtained from real marketing campaigns, while the extraction of
	propagation	probabilities and propagation delays have been well studied in the literature
	(e.g.~\cite{SKOM10,GSKZS16}).

%% file: sec_properties.tex
\vspace{-10pt}
\subsection{Properties of Influence Spread Functions} \label{sec:nonadaptive}

We now show the key properties of the three influence spread functions
	$\rho(\cdot)$, $\rho^B(\cdot)$, and $\sigma^B(\cdot)$, which are crucial for later algorithm
	design.
For a set function $f:2^V \rightarrow \R$, we say that $f$ is 
	a) {\em additive} if 
	for any subset $S\subseteq V$, $f(S) = \sum_{v\in S} f(\{v\})$;
	b) {\em monotone} if for any two subsets $S\subseteq T \subseteq V$, $f(S) \le f(T)$; and
	c) {\em submodular} if for any two subsets $S\subseteq T \subseteq V$ and an element $v\in V \setminus T$, 
		$f(T \cup \{v\}) - f(T) \le f(S\cup \{v\}) - f(S)$.
The following lemma summarizes the key properties of the three influence spread functions.

\begin{lemma}[Influence Spread Properties] \label{lem:properties}
(1) The preemptive influence spread function $\rho$ is additive;
(2) the boosted preemptive influence spread function $\rho^B$ is monotone and submodular; and
(3) the boosted influence spread function $\sigma^B$ is monotone and submodular.
\end{lemma}

%% file: sec_imm.tex
\vspace{-10pt}
\section{Scalable Implementations}\label{sec:RRset}

In this section, we develop scalable algorithms for all three problems PIM, BPIM, and BIM,
based on the reverse influence sampling (RIS) approach~\cite{BorgsBrautbarChayesLucier,tang14,tang15}.

The key concept in RIS is the reverse-reachable set.
A (random) {\em Reverse-Reachable (RR) set $R(v)$ rooted at node $v\in V$}
is  the set of nodes reachable from $v$ by reverse simulating a propagation from $v$.
More precisely, in the SAIC model, 
$R(v)$ is the set of nodes that can reach $v$ in a random possible world
$W=(A_W, \delta_W, L_W, d_W)$ following only live edges in $L_W$.
We use $\rroot(R(v))$ to denote its root $v$.
When we do not specify the root, an {\em RR set} $R$ is one 
rooted at a node picked uniformly at random from $V$.
We will use the notations $R_W(v)$ and $R_W$ when we want to clarify that the RR set is under the
	possible world $W$.

An RR set $R$ has the following intrinsic connection with the influence spread 
$\sigma(S)$ of seed set $S$ in the classical IC
model~\cite{BorgsBrautbarChayesLucier,tang14}: 
\begin{equation} \label{eq:RRsetSpread}
\sigma(S) = n\cdot \E[\I\{ S\cap R \ne \emptyset \}],
\end{equation}
where $\I$ is the indicator function.
RIS approach utilizes this fact to generate enough RR sets to estimate the influence spread and
turn influence maximization into a coverage problem of finding $k$ nodes that 
covers (a.k.a. appears in) the most number of RR sets.
We now need to adapt the RIS approach for our problems considered in the paper.
Our adaptations are patterned over the {\IMM} algorithm~\cite{tang15}, although it would be as easy 
	to adapt other state-of-the-art RIS algorithms too.

\subsection{Algorithm for BIM}\label{sec:BIM}

\begin{algorithm}[t] 
	\caption{{\IMMBIM}: adapted {\IMM} for the BIM problem
	} \label{alg:immbim}
	\KwIn{Graph $G=(V,E)$, propagation probabilities $\{p(u,v)\}_{(u,v)\in E}$,
		self-activation probabilities $\{q(u)\}_{u\in V}$, budget $k$, 
		accuracy parameters $(\varepsilon, \ell)$}
	\KwOut{seed set $S$}
	
	\tcp{Phase 1: Estimate $\theta$, the number of RR sets needed, and generate these RR sets}
	$\newalg{\cR \leftarrow \emptyset}$;  $LB \leftarrow 1$; $\varepsilon' \leftarrow \sqrt{2}\varepsilon$; $\covered\leftarrow 0$ \;
	using binary search to find a $\gamma$ such that
	$\lceil\lambda^*(\ell) \rceil /n^{\ell+\gamma} \le 1/n^\ell$ 
	\tcp{Workaround 2 in~\cite{Chen18}, $\lambda^*(\ell)$ is defined in Eq.~\eqref{eq:lambdastar}}
	$\newalg{\ell \leftarrow \ell + \gamma + \ln 2 / \ln{n}} $\;
	
	\For{$i = 1 \; to \; \log_2{(\newalg{n} - 1)}$ \label{line:bimfor1}}{
		$x_i \leftarrow \newalg{n}/2^i$\; \label{line:bimassignx}
		$\theta_i \leftarrow \lambda'/x_i$; \tcp{$\lambda'$ is defined in Eq.~\eqref{eq:lambdaprime}} \label{line:bimgenRRset1b} 
		\While{$\newalg{|\cR|} + \covered < \theta_i$}{
			Select a node $v$ from $V$ uniformly at random\; \label{line:bimsample1}
			Generate RR set $\newalg{R}$ from $v$\;
			\If{$\exists u\in R$, $u$ is self-activated with probability $q(u)$}{
				$\covered \leftarrow \covered + 1$\; \label{line:selfcount}
			}\Else{
				insert $R$ into $\newalg{\cR}$\;
			}
			\label{line:bimgsample1}
		}\label{line:bimgenRRset1e} 
		\newalg{$S_i \leftarrow \NodeSelection(\cR, k)$}\; \label{line:bimnodeselect1}
		\If{$\newalg{n} \cdot \newalg{F^S_{\cR}(S_i)} \geq (1 + \varepsilon') \cdot x_i$}{
			\tcp{$F^S_{\cR}(S)$ is defined in Eq.~\eqref{eq:fracCover}}
			\label{line:bimestimate1}
			$LB \leftarrow \newalg{n} \cdot F^S_{\cR}(S_i)/(1 + \varepsilon')$\; 
			\label{line:bimestimate2}
			{\bf break}\;  \label{line:endcheck}
		} 
	}
	$\theta \leftarrow \lambda^*(\ell) / LB$; \tcp{$\lambda^*(\ell)$ is defined in Eq.~\eqref{eq:lambdastar}}  \label{line:settheta}
	\While{$\newalg{|\cR|} + \covered \leq \theta$}{ \label{line:bimcheckLB}
		Select a node $v$ from $V$ uniformly at random\; \label{line:bimsample1}
		Generate RR set $\newalg{R}$ from $v$\;
		\If{$\exists u\in R$, $u$ is self-activated with probability $q(u)$}{
			$\covered \leftarrow \covered + 1$\; \label{line:selfcount2}
		}\Else{
			insert $R$ into $\newalg{\cR}$\;
		}
		\label{line:bimnodeselect2}		
	}
	\tcp{Phase 2: select seed nodes from the generated RR sets}
	\newalg{$S \leftarrow \NodeSelection(\cR, k)$}\;\label{line:bimgsample2}	
	\bf{return} $S$.
\end{algorithm}

%

We first present the adaption of {\IMM} to the BIM problem, since BIM is close to the
	original influence maximization problem.
The boosted influence spread $\sigma^B$ has the following connection with a random RR set $R$:

\begin{restatable}{lemma}{bim} \label{lem:bim}
	For any seed set $S$,
	\begin{equation}
	\sigma^B(S)
	= n  \cdot \E_{W\sim \cW(q,\Delta,p,D)}[\I\{ (S \cup A_W)\cap R_W \ne \emptyset \}]. \label{eq:rrbim}
	\end{equation}
\end{restatable}

Eq.~\eqref{eq:rrbim} enables the RIS approach as for the classical influence maximization.
We adapt the {\IMM} algorithm of~\cite{tang15} to get the {\IMMBIM} algorithm, as given in Algorithm~\ref{alg:immbim}.
The two main parameters $\lambda'$ and $\lambda^*(\ell)$ used in the algorithm are given below:
\begin{align}
& \lambda' \leftarrow [(2+\frac{2}{3}\varepsilon') \cdot (\ln{{\binom{n}{k}}}+\ell \cdot \ln{n}+\ln{\log_2{n}})\cdot n]/\varepsilon'^2 \label{eq:lambdaprime} \\
& \lambda^*(\ell) \leftarrow 2 n \cdot ((1-1/e) \cdot \alpha + \beta)^2 \cdot \varepsilon^{-2} 
	\label{eq:lambdastar} \\
& 	\alpha \leftarrow \sqrt{\ell \ln{n} + \ln{2}};
\beta \leftarrow \sqrt{(1-1/e)\cdot (\ln{\binom{n}{k}}+\alpha^2)} \nonumber
\end{align}

The algorithm contains two phases. 
In Phase 1, we generate $\theta$ RR sets $\cR$, where $\theta$ is computed to guarantee the approximation
	with high probability.
In Phase 2, we use the greedy algorithm to find $k$ seed nodes that
	cover as many RR sets in $\cR$ as possible: in each iteration, we find one seed node that covers
	the most number of remaining RR sets not covered by previously selected seed nodes.
The $\NodeSelection$ procedure of Phase 2 implements the above greedy algorithm, and is exactly the same as in~\cite{tang15},
	and thus we omit it here.
	
Phase 1 follows the {\IMM} structure: it uses the for-loop to estimate a lower bound of $\OPT$, the optimal solution
	to the BIM problem, by repeatedly halving the estimate $x_i$ and checking if the estimate is valid.
The validity check is by running the greedy $\NodeSelection$ procedure (line~\ref{line:bimnodeselect1})
	to find a seed set $S_i$ and getting its influence spread estimate, since the greedy algorithm should give
	a constant approximation of $\OPT$.
{\IMMBIM} differs from {\IMM} because it needs to incorporate self-activation probabilities $q(u)$'s. 
In particular, when we generate an RR set $R$, for each node $u \in R$, we sample a random coin with bias $q(u)$
	to see if $u$ is self-activated.
If so, it means this RR set $R$ has already be covered by the self-activated $u$, and there is no need to select an
	extra seed node to cover $R$.
In this case, we do not need to store $R$, but only need to count its number in variable $\covered$, which records
	the number of RR sets covered by self-activated nodes.
Only an RR set $R$ that contains no self-activated nodes needs to be stored in $\cR$ for
	later greedy seed selection (line~\ref{line:selfcount}).
The variable $\covered$ is used to estimate boosted influence spread $\sigma^B(S)$.
In particular, by Eq.~\eqref{eq:rrbim}, $\sigma^B(S)$ can be estimated as $n$ times the fraction of 
	RR sets that are covered either by $S$ or by self-activated nodes.
This fraction is defined as $F^S_{\cR}(S)$:
\begin{equation} \label{eq:fracCover}
F^S_{\cR}(S) = \frac{\covered + \sum_{R\in \cR} \I\{S\cap R \ne \emptyset \}}{\covered + |\cR|}.
\end{equation}
By an analysis similar to that of the \alg{IMM} algorithm~\cite{tang15}, we have
	the following theorem.
Let $\tilde{v}$ be a random node selected from $V$ with probability proportional to its indegree, and
let $\sigma(\tilde{v})$ denote the influence spread of $\tilde{v}$ in the corresponding IC model.

\begin{restatable}{theorem}{bimw} \label{thm:bimw}
	Let $S^*$ be the optimal solution of the BIM problem.
	For every $\varepsilon > 0$ and $\ell > 0$, with probability at least
	$1-\frac{1}{n^\ell}$, the output $S^o$ of {\IMMBIM}
	satisfies
	\begin{equation*}
		\sigma^B(S^o) \geq \left(1- \frac{1}{e} - \varepsilon\right) \sigma^B(S^*),
	\end{equation*}
	In this case, the expected running time for {\IMMBIM}
	is $O((k + \ell)(n + m)\log{n}/\varepsilon^2\cdot (\E[\sigma(\tilde{v})]/ \sigma^B(S^*))) = O((k + \ell)(n + m)\log{n}/\varepsilon^2)$.
\end{restatable}
Similar to {\IMM}, the above theorem shows that {\IMMBIM} achieves $1-1/e-\varepsilon$
	approximation with near-linear running time.
The theorem explicitly shows the ratio $\E[\sigma(\tilde{v})]/ \sigma^B(S^*)$, 
	which is less than $1$, 
	in order to compare later with other algorithms.


\vspace{-5pt}
\subsection{Algorithm for BPIM}\label{sec:BPIM}

We next discuss our implementation of BPIM, which we call {\IMMBPIM}.
Since the objective function $\rho^B(S)$ is monotone submodular (Lemma~\ref{lem:properties}), 
	{\IMMBPIM} follows the general structure of greedy seed selection.
However, it differs from {\IMM} and {\IMMBIM} significantly in its RR set definition and generation process.
Intuitively, the preemptive influence spread of a seed node $u$ only counts the activated nodes that
	$u$ reaches first before any other seed nodes or self-activated nodes.
In terms of RR sets, this means that a node $u$ can be included in the preemptive RR set only if
	the total delay from $u$ to the root $v$ is smaller than the minimum total delay from any
	self-activated nodes to $v$.
We formally define the {\em preemptive reverse-reachable (P-RR) set as follows}.
Given a possible world $W=(A_W, \delta_W, L_W, d_W)$ in the SAIC model, a P-RR set $R^P_W(v)$ rooted at $v$ is the
	set of nodes $u$ such that (1) $u$ could reach $v$ through live edges in $L_W$, and
	(2) the total delay of $u$ to $v$, $T_W(u,v)$, is less than or equal to the minimum total delay from 
	any self-activated node in $A_W$ to $v$.
When we do not specify the root $v$, P-RR set $R^P_W$ is a P-RR set with root selected uniformly
	at random among all nodes in $V$.
The subscript $W$ could be omitted when the context is clear.
With this definition, we can obtain the following connection between a P-RR set and the preemptive
	influence spread $\rho^B(S)$.

\begin{restatable}{lemma}{bpim} \label{lem:bpim}
	For any seed set $S$,
	\begin{equation}
	\rho^B(S)
	= n  \cdot \E_{W\sim \cW(q,\Delta,p,D)}[\I\{ S\cap R^P_W \ne \emptyset \}]. \label{eq:rrbpim}
	\end{equation}
\end{restatable}

\begin{algorithm}[t] 
	\caption{{\PRR}: Preemptive RR Set Generation} \label{alg:pwrr}
	\KwIn{root $v^r$, Graph $G=(V,E)$, self-activation probability $p$, random distribution of self-activation delay $\Delta$, propagation probability $p$, random distribution of propagation delay $D$}
	\KwOut{P-RR set $R^P$, node $u^s$ that is the first activating $v^r$}
	
	$Q \leftarrow \{v^r\}$; $R^P \leftarrow \emptyset$; $u^s \leftarrow -1$ \;
	\For{each $v \in V$}{
		$\delay[v] \leftarrow +\infty$;\tcp{initial delays for reaching the root}
	}
	$\delay[v^r] \leftarrow 0$\;
	\While{$Q \neq \emptyset$}{
		$w \leftarrow \argmin_{w'\in Q} \delay[w']$\; \label{line:selectmin}
		delete $w$ from $Q$\;
		\If{$w$ is a shadow node $\shadow{v}$}{
			insert $v$ into $R^P$\; \label{line:insertRP}
			\If{$v$ is self-activated with probability $q(v)$ \label{line:foundminnodeb}}{
				$u^s \leftarrow v$\;
				{\bf break}\; \label{line:foundminnodee}
			}
		}\Else{
			\tcp{let $w$ be the real node $v$}
			sample $\delta(v) \sim \Delta(v)$\; \label{line:samplesadelay}
			$\delay[\shadow{v}] \leftarrow \delay[v] + \delta(v)$\;
			insert $\shadow{v}$ into $Q$\; \label{line:insertshadow}
			\For{each real in-neighbor $u$ of $v$ in $G$ \label{line:reversesimu}}{
				\If{$(u,v)$ is sampled as live with probability $p(u,v)$}{
					sample $d(u,v) \sim D(u,v)$\;
					$tmp \leftarrow \delay[v] + d(u,v)$\;
					\If{ $\delay[u] = +\infty$}{
						insert $u$ into $Q$\;
					}
					\If{$tmp < \delay[u]$ \label{line:trimexplore}
					}{
						$\delay[u] \leftarrow tmp$ \label{line:updatepd}\; \label{line:updatedelay}
					}
				}
			}
		}
	}
	\bf{return} $R^P$, $u^s$.
\end{algorithm}

With Eq.~\eqref{eq:rrbpim}, we can see that as long as we can properly generate P-RR sets, 
	we can use the {\IMM} algorithm in the same way to find the seed sets.
Therefore, our main focus now is to efficiently generate a random P-RR set.
This algorithm is implemented as {\PRR} as given in Algorithm~\ref{alg:pwrr}.
The main idea is that for each node $v\in V$, we add its shadow node $\shadow{v}$ and an edge from $\shadow{v}$ to $v$, 
	and consider the delay on the edge $(\shadow{v},v)$ as a sample of the self-activation delay $\delta(v)\sim \Delta(v)$.
Then the delay from any node (real or shadow) to the root $v^r$ is the minimum delay along any path from the node to $v^r$.
Let $u^s$ be a node that is self-activated and its corresponding shadow node $\shadow{v}^s$ has the minimum delay to the root $v^r$ among
	all shadow nodes.
Then the P-RR set $R^P$ is the set of nodes whose shadow nodes have delay less than or equal to the delay of $\shadow{v}^s$.

To find $u^s$ and $R^P$, we apply the idea from the Dijkstra's shortest path algorithm: from the candidate nodes we touched so far
	(set $Q$), we take
	node $w$ that has the shortest delay as the next one to explore (line~\ref{line:selectmin}). 
If $w$ is a shadow node $\shadow{v}$, we insert $v$ into $R^P$ (line~\ref{line:insertRP}), then 
	test if $v$ is self-activated or not, and if so, we find $u^s = v$, and $R^P$ contains
	all the shadow nodes we have explored so far and the algorithm stops (lines~\ref{line:foundminnodeb}--\ref{line:foundminnodee}).
If $w$ is a real node $v$, we first sample the delay $\delta(v)\sim \Delta(v)$ as the delay from $v$'s shadow $\shadow{v}$ to $v$, 
	compute
	the delay of $\shadow{v}$ as $\delay[v] + \delta(v)$, and insert $\shadow{v}$ into the candidate node set $Q$
	(lines~\ref{line:samplesadelay}--\ref{line:insertshadow}).
We then do reverse simulation along all of $v$'s incoming edges $(u,v)$, and sample the edge delay $d(u,v)\sim D(u,v)$, and do proper
	updates of $\delay[u]$ (lines~\ref{line:reversesimu}--\ref{line:updatedelay}).
The algorithm guarantees that the node sequence we explore has increasing delays, which in turn guarantees the correctness of
	$R^P$ and $u^s$ found by the algorithm.
We remark that $u^s$ would be useful in solving PIM, as to be explained in the next subsection.

With the P-RR set generation algorithm {\PRR}, we just plug it into the {\IMM} algorithm and obtain 
	{\IMMBPIM}.
The full pseudocode is omitted.
We have the theorem below for the {\IMMBPIM} algorithm.

\begin{restatable}{theorem}{thmIMMBPIM} \label{thm:IMMBPIM}
	Let $S^*$ be the optimal solution of the BPIM.
	For every $\varepsilon > 0$ and $\ell > 0$, with probability at least
	$1-\frac{1}{n^\ell}$, the output $S^o$ of the {\IMMBPIM} algorithm 
	satisfies
	\begin{equation*}
	\rho^B(S^o) \geq \left(1- \frac{1}{e} - \varepsilon\right) \rho^B(S^*).
	\end{equation*}
	In this case, 
	the expected running time of the {\IMMBPIM} algorithm 
	is $O((k + \ell)(n + m)\log^2{n}/\varepsilon^2\cdot (\E[\sigma(\tilde{v})]/ \rho^B(S^*)))$.
\end{restatable}

Note that for the ratio $\E[\sigma(\tilde{v})]/ \rho^B(S^*)$, one would expect that typically the optimal solution of BPIM would
	be larger than any single node influence spread, and thus the ratio is less than $1$ and we have a near-linear time algorithm.
In this case, the expected running time for {\IMMBPIM} still has one extra $\log n$ term comparing to that of {\IMMBIM} or {\IMM}.
This is because our reverse simulation algorithm {\PRR} needs to run a Dijkstra-like algorithm, and in particular
	we need to implement the set $Q$ in Algorithm~\ref{alg:pwrr} as a priority queue to 
	support insertion, deletion, updates, and finding the minimum
	value.

\begin{algorithm}[t] 
	\caption{{\IMMPIM}: Preemptive IMM Algorithm} \label{alg:immpim}
	\KwIn{Graph $G=(V,E)$, self-activation probabilities $q$, self-activation delay distributions $\Delta$, propagation probabilities $p$, 
		propagation delay distributions $D$, budget $k$, accuracy parameters $(\varepsilon, \ell)$
	}
	\KwOut{set $S$}
	
	\tcp{Phase 1: Estimate $\theta$, the number of P-RR sets needed, and generate these P-RR sets}
	$\newalg{\cR \leftarrow \emptyset}$;  $LB \leftarrow 1$; $\varepsilon' \leftarrow \sqrt{2}\varepsilon$\;
	using binary search to find a $\gamma$ such that
	$\lceil\tilde{\lambda}^*(\ell) \rceil /n^{\ell+\gamma} \le 1/n^\ell$ 
	\tcp{Workaround 2 in~\cite{Chen18}, $\tilde{\lambda}^*(\ell)$ is defined in Eq.~\eqref{eq:tildelambdastar}}
	$\newalg{\ell \leftarrow \ell + \gamma + \ln 2 / \ln{n}} $\;

	$est_u \leftarrow 0$ for every $u \in V$\;
	\For{$i = 1$ to $\lfloor\log_2{n}\rfloor - 1$}{
		$x_i \leftarrow \newalg{n}/2^i$\; \label{line:assignx}
		$\theta_i \leftarrow \lambda'/x_i$; \tcp{$\lambda'$ is defined in Eq.~\eqref{eq:lambdaprime}}
		\While{$\newalg{|\cR|} \leq \theta_i$}{
			Select a node $v$ from $V$ uniformly at random\;\label{line:pimsample1}
			$(-, u) \leftarrow \PRR(v,G,q,\Delta,p,D)$;
			\tcp{generate a random P-RR set pair $(R^P,u)$, need the
				returned node $u$ that is both self-activated and the earliest in reaching root $v$,
				but ignore the set $R^P$} \label{line:pimgenRRset1a}
			\If{$u \ne -1$}{
				$est_{u} \leftarrow est_{u} + 1$\;
				\label{line:pimupdate1}
			}
		}
		$\topk \leftarrow $ sum of the top $k$ largest values in $\{est_u\}_{u \in V}$\;
		\If{$n \cdot \topk \geq (1 + \epsilon') \cdot x_i$}{
			\label{line:pimestimate1}
			$LB \leftarrow \newalg{ n \cdot \topk/(\theta_i \cdot (1 + \epsilon'))}$\; 
			\label{line:pimestimate2}
			{\bf break}\;
		} 
	}
	$\theta \leftarrow \tilde{\lambda}^* / LB$; \tcp{$\tilde{\lambda}^*$ is defined in Eq.~\eqref{eq:tildelambdastar}}
	\While{$\newalg{|\cR|} \leq \theta$}{
		Select a node $v$ from $V$ uniformly at random\;\label{line:pimsample2}
		$(-, u) \leftarrow \PRR(v,G,q,\Delta,p,D)$\;\label{line:pimgenRRset1b}
		\If{$u \ne -1$}{
			$est_{u} \leftarrow est_{u} + 1$\;
			\label{line:pimupdate2}
		}
	}
	\tcp{Phase 2: obtain the top $k$ nodes}
	$S \leftarrow$ set of top $k$ nodes with the largest values in $\{est_u\}_{u \in V}$\;
	\label{line:pimoutput}
	\bf{return} $S$.
\end{algorithm}

\vspace{-5pt}
\subsection{Algorithm for PIM}\label{sec:PIM}

Finally, we consider the preemptive influence maximization (PIM) algorithm.
PIM differs from BIM and BPIM in that we do not select seeds and boost their self-activation probabilities to $1$.
We only select $k$ nodes who {\em spontaneously} have the largest preemptive influence spread, due to self activations.
By Lemma~\ref{lem:properties}, we know that the preemptive influence spread $\rho$ is additive, which implies that
	we just need to estimate individual node's preemptive influence spread and select the top $k$ of them.
We still use the RIS approach for estimating individual node's preemptive influence spread, based on the following result.
Let $u^s_W(v)$ be the node in the possible world $W$ that is self-activated and can reach $v$ with the minimum total delay $T(u,v)$, and
	$u^s_W$ denotes such a random $u^s_W(v)$ where $v$ is selected uniformly at random.

\begin{restatable}{lemma}{lempim} \label{lem:pim}
	For any node $u$,
	\begin{equation}
	\rho(\{u\}) = n \cdot \E_{W\sim \cW(q,\Delta,p,D)}[\I\{ u = u^s_W\}]. \label{eq:rrpim}
	\end{equation}
\end{restatable}

With Lemma~\ref{lem:pim}, we can randomly select a root $v$, and reverse simulate from $v$ to find the node
	$u^s_W(v)$, and for each such node $u = u^s_W(v)$, we compute the fraction of times it appears in the reverse simulation, 
	and multiply it with $n$ to get $u$'s preemptive influence spread $\rho(\{u\})$.
This reverse simulation procedure has been done as part of {\PRR} algorithm, and its output $u^s$ is the
	$u^s_W$ we refer here.

With the above new way of reverse simulation, we can plug it into the {\IMM} framework to obtain our algorithm {\IMMPIM}~\ref{alg:immpim}.
The main difference is that we do not need greedy {\NodeSelection} procedure to give a $1-1/e$ approximation of the optimal seed
	set covering the RR set sequence $\cR$.
Instead, each node $u$ maintains a counter $est_u$ to record the number of times the reverse simulation hits $u^s = u$, and we just select
	the top $k$ nodes with the largest counters as our output set.
This would give a $1-\varepsilon$ approximation instead of the $1-1/e - \varepsilon$ approximation as previous algorithms.
For the same reason, the parameter $\lambda^*(\ell)$ should be redefined, replacing the factor $1-1/e$ in the parameter with $1$.
With these changes, {\IMMPIM} has the following theoretical guarantee.
\begin{restatable}{theorem}{thmIMMPIM} \label{thm:IMMPIM}
	Let $S^*$ be the optimal solution of the PIM.
	For every $\varepsilon > 0$ and $\ell > 0$, with probability at least
	$1-\frac{1}{n^\ell}$, the output $S^o$ of the {\IMMBPIM} algorithm 
	satisfies
	\begin{equation*}
		\rho(S^o) \geq \left(1- \varepsilon\right) \rho(S^*).
	\end{equation*}
	In this case, 
	the expected running time of the {\IMMPIM} algorithm 
	is $O((k + \ell)(n + m)\log^2 n/\varepsilon^2 \cdot (\E[\sigma(\tilde{v})]/ \rho(S^*)))$.
\end{restatable}

%

%% file: sec_exp.tex
\section{Empirical Evaluation} \label{sec:experiment}

The main purpose of our empirical evaluation is to validate if and when using our self-activation 
aware influence maximization algorithms are beneficial over using the classical
self-activation oblivious algorithms, and to quantify is the difference in performance.
We conduct experiments on two real-world social networks to test the performance 
of our algorithms and compare them with the classical influence maximization and the Shapley centrality algorithms.

\begin{figure*}[tb]
	\centering
	\subfloat[PIM-NetHEPT]{\includegraphics[width=1.38in]{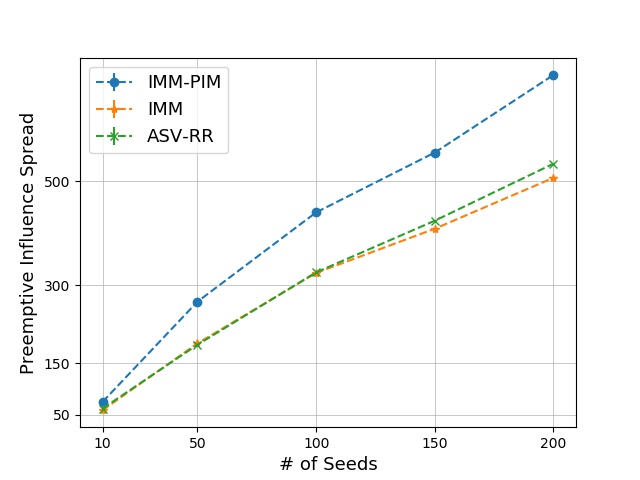}}\quad
	\subfloat[PIM-Flixster]{\includegraphics[width=1.38in]{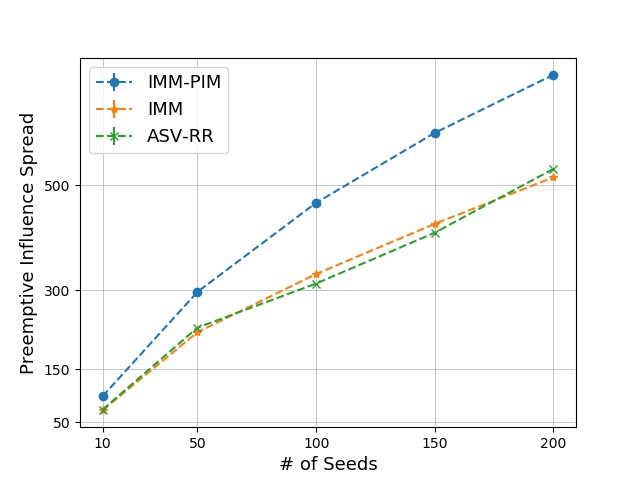}}\quad
	\subfloat[BPIM-NetHEPT]{\includegraphics[width=1.38in]{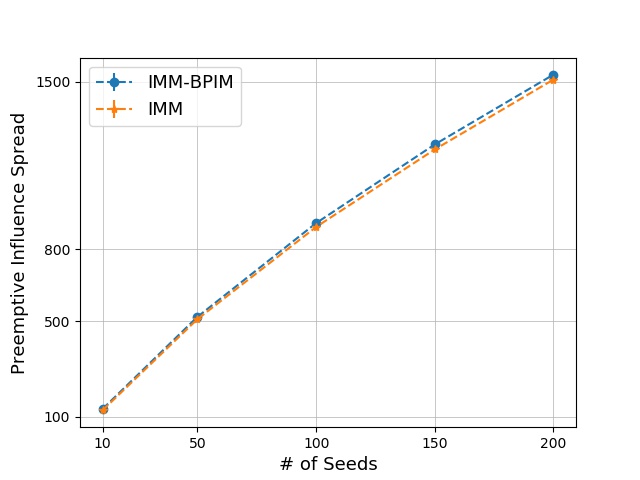}}\quad
	\subfloat[BPIM-NetHEPT]{\includegraphics[width=1.38in]{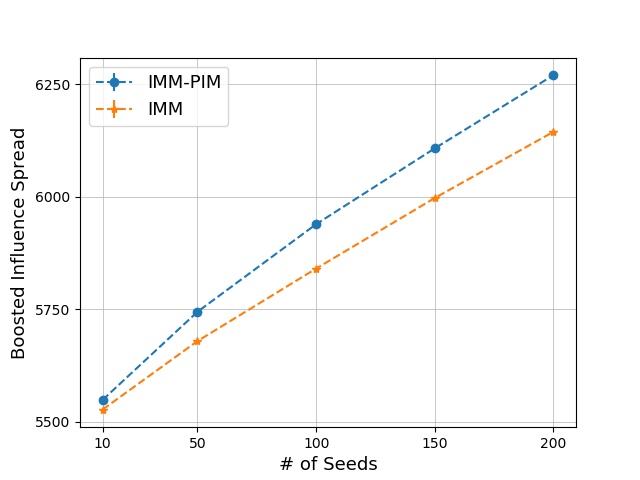}}
	\caption{Influence Spread Results ($\varepsilon = 0.1$ with case 3).} \label{fig:infscore}
\end{figure*}

\begin{table*}[t!]
	\centering
	\caption{Running time results (in seconds).} \label{tab:time1}
	\resizebox{6.0in}{!}{
		\begin{tabular}{|l|c|c|c|c|c|c|l|c|c|c|c|c|}
			\cline{1-6} \cline{8-13}
			Data    &  \IMMBIM & \IMMPIM & \IMMBPIM & \IMM    & \ASVRR &  & Data    &  \IMMBIM & \IMMPIM & \IMMBPIM & \IMM    & \ASVRR \\ \cline{1-6} \cline{8-13} 
			NetHEPT & 0.7534  & 195.75  & 48.123   & 1.9712 & 74.175 &  & Flixster & 1.3516 & 955.45  & 218.51   & 5.1072 & 235.34 \\ \cline{1-6} \cline{8-13}
		\end{tabular}
	}
\end{table*}

\subsection{Experiment Setup}

{\bf Data Description.\ \ } We use the following two datasets, all of which have been used in a number of influence
maximization studies.
(1) \textbf{Flixster}.
The Flixster dataset is a network of American social movie discovery service (www.flixster.com). 
To transform the dataset into a weighted graph, 
each user is represented by a node, 
and a directed edge from node $u$ to $v$ is formed if $v$ rates one movie shortly 
after $u$ does so on the same movie. 
The dataset is analyzed in \cite{barbieri2012topic}, 
and the influence probability are learned by the topic-aware model. 
We use the learning result of \cite{barbieri2012topic} in our experiment, 
which is a graph containing \num{29357} nodes and \num{212614} directed edges. 
There are 10 probabilities on each edge, 
and each probability represents the influence
from the source user to the sink on a specific topic. 
In our experiment, we test the first topic.
(2) \textbf{NetHEPT}. 
The NetHEPT dataset \cite{ChenWY09} is extensively used in many influence maximization studies. 
It is an academic collaboration network from the ``High Energy Physics Theory'' section of arXiv from 1991 to 2003, 
where nodes represent the authors and each edge represents one paper co-authored by two nodes. 
There are \num{15233} nodes and \num{58891} undirected edges (including duplicated edges) in the NetHEPT dataset.
We clean the dataset by removing those duplicated edges and obtain a directed graph $G = (V,E)$,
$|V|$ = \num{15233}, $|E|$ = \num{62774} (directed edges). 
The propagation probability on edges are set by weighted cascade model~\cite{kempe03}:
	the probability of an edge $(u,v)$ is set as the inverse of the in-degree of $v$.
(3) \textbf{DBLP}. 
The DBLP dataset \cite{WCW12} is an academic collaboration network extracted from online archive
DBLP (dblp.uni-trier.de). There are 654K nodes and 1990K directed edges in the DBLP.
The propagation probabilities on the edges are also set by the weighted cascade model.

\noindent
{\bf Algorithms. \ \ }
For the two problems (PIM, BPIM), we test our corresponding algorithms with
the baseline {\IMM}, which is oblivious to the self-activation behaviors and treat
all nodes as no self activation and seed nodes as activated at time $0$.
For the PIM problem, we further compare against the efficient Shapley computation
algorithm {ASV-RR} proposed in~\cite{chen2017interplay}, which essentially treats
all nodes as self-activations in a uniform random order.
We use the same parameters settings for these algorithms: $\ell = 1$, $\varepsilon = 0.1$.
We test seed set sizes of $10, 50, 100, 150$ and $200$.

\noindent
{\bf Self-activation parameters and test cases. \ \ }
In practice, self-activation delays can be estimated from the users' access patterns to online
social networks, and self-activation probabilities can be estimated from the
fraction of times users' participating in information cascades not due to the influence
from the neighbors or external selections as seed users.
Unfortunately for the datasets we use, these information are not available.
Instead, we use synthetic settings, focusing on whether the knowledge of the self-activation
behaviors would benefit our algorithm design.
For self-activation probability $q(u)$ of node $u$, we 
first randomly select a value $\beta_u$ from $[0,c]$ as
a node $u$'s base value, then we further test five cases:
(0) uniform: $q(u) = \alpha_u$; 
(1) positively correlated: $q(u)$ is positively correlated with $u$'s out-degree
$d^+(u)$, in particular $q(u) = \min\{\beta_u \cdot d^+(u), 1\}$;
(2) negatively correlated: $q(u) = \beta_u / d^+(u)$;
(3) random mixing of cases 0 and 1: randomly pick half of the nodes with 
$q(u) = \alpha_u$ and the other half with $q(u) = \min\{\beta_u \cdot d^+(u), 1\}$;
(4) random mixing of cases 0 and 2: randomly pick half of the nodes with 
$q(u) = \alpha_u$ and the other half with $q(u) = \beta_u / d^+(u)$.
These five cases aim at scenarios where all users are equally likely to react to
a campaign (case 0), high-degree nodes (usually more influential) are more likely
or less likely to react to the campaigns, and mixture of uniform behavior and a 
correlation (or reverse correlation) behavior.
We set $c = 2$ for BPIM and PIM tests,
because otherwise the preemptive influence spread for PIM is too small.
For self-activation delays, we use exponential distribution with rate $1$ for all nodes.
We already vary the self-activation behaviors through the self-activation probabilities,
and thus we simply keep the self-activation delay distributions uniform.
We also use the same exponential distribution for propagation delay distributions.
Two proposed algorithms and two baselines are written in c++ and compiled by Visual studio.
All experiments are conducted  
on a 15'' MacBook Pro with a \SI{2.5}{GHz} Intel Core~i7 and
\SI{16}{GB} of \SI{1600}{MHz} DDR3 memory. 


\vspace{-10pt}
\subsection{Results}


\noindent
{\bf Influence spread result. \ \ } As we can see in Figure~\ref{fig:infscore},
\IMMPIM\ significantly outperforms the baselines on the achieved preemptive influence spread,
and \IMMBPIM\ also outperforms others when self-activation behaviors of the nodes
are non-uniform.
On average \IMMPIM\ improves about 32.7\% than the other baselines, \IMMBPIM\ improves about 2.3\% than \IMM\ algorithm, and \IMMBIM\ improves about 2.1\% than \IMM\ algorithm..
Detailed results are in extended version \cite{sun2019self}.

\noindent
{\bf Running time results. \ \ }
Table~\ref{tab:time1} reports the running time of all algorithms on both datasets,
by using the default setting with the seed set size $k = 200$ in test case 3.
We can clearly see the order of running time is $\BIM < \IMM < \IMMBPIM < \ASVRR < \IMMPIM$
(we ignore the prefix {\IMM} in our algorithm to fit the table width).
This is inline with our theoretical analysis, which shows that the running time is
inversely proportional to the optimal value of each problem.
For example, the optimal solution of BIM is larger than that of the classical 
	influence maximization because BIM has self-activated nodes contributing extra influence spread, 
	and PIM has the smallest optimal value because 
the self-activation probabilities are small in general and the optimal set has
to compete with other self-activated nodes on preemptive influence spread.
{\IMMBPIM} and {\IMMPIM} are further slower due to the Dijkstra-like reverse simulation, which
takes more time than the simple breadth-first-search simulation.
But even for the slowest {\IMMPIM} algorithm, on the Flixster dataset with more than hundred thousands nodes and edges, it could complete in less than 16 minutes on our laptop test machine.
Besides the dataset, running time is also related to the $\varepsilon$.
We also test the $\varepsilon$ values from 0.1 to 0.5 on NetHEPT for the proposed algorithms, and the details are in the extented version due to the space constraint \cite{sun2019self}.



In summary, our tests clearly demonstrate that for the PIM problem targeted for identifying
top organic influencers in a graph with self-activation, our {\IMMPIM} algorithm significantly
outperform other baselines in terms of result quality, which suggests that knowing
the self-activation behavior is important for this task.
For BPIM, our algorithm have small improvements for certain test cases with
non-uniform self-activation behaviors.
This may suggest that baseline such as {\IMM} may be usable for these tasks, but one could
still benefit from our algorithms in certain cases.

\vspace{-5pt}

\section{Conclusion and Future Work}

We introduce self activation and preemptive influence maximization task in this study.
A future direction of our study would be incorporating self activation and preemptive
influence spread considerations into other influence maximization tasks, 
such as competitive and complementary influence maximization, adaptive and
online influence maximization, etc.

\vspace{-5pt}
\section{Acknowledgments}
This work is supported in part by NSF under grants III-1526499, III-1763325, III-1909323, CNS-1930941, and CNS-1626432.

%% file: appendix.tex
\appendix

\section{Pseudocode for {\IMMPIM}}

The parameter $\lambda^*(\ell)$ is replaced by $\tilde{\lambda}^*(\ell)$, defined as follows:
we redefine it as $\tilde{\lambda}^*(\ell)$ below:
\begin{align}
	& \tilde{\lambda}^*(\ell) \leftarrow 2 n \cdot (\alpha + \tilde{\beta})^2 \cdot \varepsilon^{-2} 
	\label{eq:tildelambdastar} \\
	& 	\alpha \leftarrow \sqrt{\ell \ln{n} + \ln{2}};
	\tilde{\beta} \leftarrow \sqrt{\ln{\binom{n}{k}}+\alpha^2}.\nonumber
\end{align}

The pseudocode of {\IMMPIM} is given in Algorithm~\ref{alg:immpim}.
The main difference comparing with {\IMM} or {\IMMBPIM} is that:
	(a) replacing the greedy {\NodeSelection} procedure by simply counting
		the number of occurrences of each node $u$ as the first node reaching $v$
		in the reverse simulation (through variable $est_u$) and selecting the
		top $k$ nodes with the highest number of occurrences; and
	(b) replacing $\lambda^*(\ell)$ with $\tilde{\lambda}^*(\ell)$.

\begin{algorithm}[t] 
	\caption{{\IMMPIM}: Preemptive IMM Algorithm} \label{alg:immpim}
	\KwIn{Graph $G=(V,E)$, self-activation probabilities $q$, self-activation delay distributions $\Delta$, propagation probabilities $p$, 
		propagation delay distributions $D$, budget $k$, accuracy parameters $(\varepsilon, \ell)$
	}
	\KwOut{set $S$}
	
	\tcp{Phase 1: Estimate $\theta$, the number of P-RR sets needed, and generate these P-RR sets}
	$\newalg{\cR \leftarrow \emptyset}$;  $LB \leftarrow 1$; $\varepsilon' \leftarrow \sqrt{2}\varepsilon$\;
	using binary search to find a $\gamma$ such that
	$\lceil\tilde{\lambda}^*(\ell) \rceil /n^{\ell+\gamma} \le 1/n^\ell$ 
	\tcp{Workaround 2 in~\cite{Chen18}, $\tilde{\lambda}^*(\ell)$ is defined in Eq.~\eqref{eq:tildelambdastar}}
	$\newalg{\ell \leftarrow \ell + \gamma + \ln 2 / \ln{n}} $\;

	$est_u \leftarrow 0$ for every $u \in V$\;
	\For{$i = 1$ to $\lfloor\log_2{n}\rfloor - 1$}{
		$x_i \leftarrow \newalg{n}/2^i$\; \label{line:assignx}
		$\theta_i \leftarrow \lambda'/x_i$; \tcp{$\lambda'$ is defined in Eq.~\eqref{eq:lambdaprime}}
		\While{$\newalg{|\cR|} \leq \theta_i$}{
			Select a node $v$ from $V$ uniformly at random\;\label{line:pimsample1}
			$(-, u) \leftarrow \PRR(v,G,q,\Delta,p,D)$;
			\tcp{generate a random P-RR set pair $(R^P,u)$, need the
				returned node $u$ that is both self-activated and the earliest in reaching root $v$,
				but ignore the set $R^P$} \label{line:pimgenRRset1a}
			\If{$u \ne -1$}{
				$est_{u} \leftarrow est_{u} + 1$\;
				\label{line:pimupdate1}
			}
		}
		$\topk \leftarrow $ sum of the top $k$ largest values in $\{est_u\}_{u \in V}$\;
		\If{$n \cdot \topk \geq (1 + \epsilon') \cdot x_i$}{
			\label{line:pimestimate1}
			$LB \leftarrow \newalg{ n \cdot \topk/(\theta_i \cdot (1 + \epsilon'))}$\; 
			\label{line:pimestimate2}
			{\bf break}\;
		} 
	}
	$\theta \leftarrow \tilde{\lambda}^* / LB$; \tcp{$\tilde{\lambda}^*$ is defined in Eq.~\eqref{eq:tildelambdastar}}
	\While{$\newalg{|\cR|} \leq \theta$}{
		Select a node $v$ from $V$ uniformly at random\;\label{line:pimsample2}
		$(-, u) \leftarrow \PRR(v,G,q,\Delta,p,D)$\;\label{line:pimgenRRset1b}
		\If{$u \ne -1$}{
			$est_{u} \leftarrow est_{u} + 1$\;
			\label{line:pimupdate2}
		}
	}
	\tcp{Phase 2: obtain the top $k$ nodes}
	$S \leftarrow$ set of top $k$ nodes with the largest values in $\{est_u\}_{u \in V}$\;
	\label{line:pimoutput}
	\bf{return} $S$.
\end{algorithm}

\vspace{-5pt}
\section{Proofs of Theorems and Lemmas}

\begin{proof}[Proof of Lemma~\ref{lem:properties}]
	We prove all three statements on a fixed possible world $W=(A_W, \delta_W, L_W, d_W)$, 
	since taking expectation on $W$ preserves
	additivity, monotonicity, and submodularity.
	For (1), it is straightforward if we observe that for each activated node $v$ in $W$, there is
	a unique source node $u \in A_W$ that is the first activating $v$, by our assumption that
	no two live paths in $W$ have the same total delay. 
	This implies that $\Gamma_W(\{u\}) \cap \Gamma_W(\{u'\}) = \emptyset$ for any two different 
	nodes $u,u' \in V$, 
	and thus $|\Gamma_W(A)| = \sum_{u\in A} \Gamma_W(\{u\})$, for any $A \subseteq V$.
	For (2), monotonicity is trivial.
	For submodularity, it is sufficient to prove that for any two subsets $S\subseteq T \subseteq V$ 
	and $u \in V \setminus T$, 
	$\Gamma^B_W(T\cup \{u\}) \setminus \Gamma^B_W(T) \subseteq 
	\Gamma^B_W(S\cup \{u\}) \setminus \Gamma^B_W(S) $.
	For a node $v \in \Gamma^B_W(T\cup \{u\}) \setminus \Gamma^B_W(T)$, 
	there must exist a live path from $u$ to $v$ in $W$, such that 
	the total delay of $P$, $T_W(P)$, is the minimum among all live paths from $T \cup \{u\} \cup A_W$
	to $v$, which also implies that $T_W(P)$ is strictly less than the total delay of any live path
	from $T \cup A_W$  to $v$.
	Since $S \subseteq T$, this directly implies that $v \in \Gamma^B_W(S\cup \{u\})$
	but $v \not\in \Gamma^B_W(S) $.
	For (3), again the monotonicity is trivial, and submodularity proof is similar to (2).
	It is sufficient to prove that for any two subsets $S\subseteq T \subseteq V$ 
	and $u \in V \setminus T$, 
	$\Phi^B_W(T\cup \{u\}) \setminus \Phi^B_W(T) \subseteq 
	\Phi^B_W(S\cup \{u\}) \setminus \Phi^B_W(S) $.
	For a node $v \in \Phi^B_W(T\cup \{u\}) \setminus \Phi^B_W(T)$, 
	there is no live path from any node in $T\cup A_W$ to $v$ but 
	there exists a live path from $u$ to $v$ in $W$.
	Since $S \subseteq T$, this directly implies that $v \in \Phi^B_W(S\cup \{u\})$
	but $v \not\in \Phi^B_W(S) $.
\end{proof}

\begin{proof}[Proof of Lemma~\ref{lem:bim}]
	\begingroup
	\allowdisplaybreaks
	\begin{align}
		& \E_{W\sim \cW(q,\Delta,p,D)}[\I\{ (S \cup A_W)\cap R_W \ne \emptyset \}] \nonumber \\
		& = \sum_{v\in V} \Pr\{v = \rroot(R_W)\} \cdot 
		\E_{W}[\I\{ (S \cup A_W)\cap R_W \ne \emptyset \mid \nonumber \\ 
		& \qquad \qquad	 v = \rroot(R_W)\}] 
		\nonumber \\
		& = \frac{1}{n} \sum_{v\in V} \E_{W}[\I\{ (S \cup A_W)\cap R_W(v) \ne \emptyset \} ]
		\nonumber \\
		& = \frac{1}{n} \sum_{v\in V} \E_W[\I\{ v\in \Phi^B_W(S) \} ] \label{eq:bimimm} \allowdisplaybreaks \\
		& = \frac{1}{n} \cdot \E_W[|\Phi^B_W(S) | ] \nonumber \\
		& =  \frac{1}{n} \cdot \sigma^B(S), \nonumber 
	\end{align}
	\endgroup
	%
	%
	%
	%
	where Eq.~\eqref{eq:bimimm} is by the definitions of RR set $R_W(u)$ and final activated set $\Phi^B_W(S)$.
\end{proof}

\begin{proof}[Proof of Theorem~\ref{thm:bimw} (Sketch)]
The proof follows the same structure as the proof of {\IMM} in~\cite{tang15} together with the workaround 2 
	summarized in~\cite{Chen18}.
Let $\cR$ be the sequence of RR sets where every $R\in \cR$ has no self-activated nodes, and thus $\cR$ is the sequence
	generated by the {\IMMBIM} algorithm (Algorithm~\ref{alg:immbim}).
Let $\cR^S$ be the sequence of RR sets where every $R\in \cR^S$ contains some node that is self-activated.
By Lemma~\ref{lem:bim}, we can have an unbiased estimate of boosted influence spread as
	$\hat{\sigma}^B(S) = n \cdot (\sum_{R\in \cR} \I\{ S\cap R \ne \emptyset \} + |\cR^S|)/(|\cR| + |\cR^S|) $.
Note that $|\cR^S|$ is exactly maintained by variable $\covered$ in Algorithm~\ref{alg:immbim}, therefore, the above formula
	matches the definition of $F^S_{\cR}(S)$ in Eq.\eqref{eq:fracCover}.
Also notice that only for RR sets not containing self-activated nodes, we need to find seeds to cover them, and thus
	the procedure {\NodeSelection} only takes $\cR$ as the input.
The rest of the analysis would follow the same way as the {\IMM} analysis, since the boosted influence spread
	is monotone submodular (Lemma~\ref{lem:properties}), same as the classical influence spread.
\end{proof}

\begin{proof}[Proof of Lemma~\ref{lem:bpim}]
	\begin{align}
		& \E_{W\sim \cW(q,\Delta,p,D)}[\I\{ S \cap R^P_W \ne \emptyset \}]  \nonumber \\
		& = \sum_{v\in V} \Pr\{v = \rroot(R^P_W)\} \cdot \E_W[\I\{ S \cap R^P \ne \emptyset \mid v = \rroot(R^P_W)\}]
		\nonumber \\
		& = \frac{1}{n} \sum_{v\in V} \E_W[\I\{ S\cap R^P_{W}(v) \ne \emptyset \}]
		\nonumber \\
		& = \frac{1}{n} \sum_{v\in V}\E_W[ \I\{v \in \Gamma^B_W(S)\} ] \label{eq:bpimimm} \allowdisplaybreaks \\
		& = \frac{1}{n} \cdot \E_W[|\Gamma^B_W(S)| ] \nonumber \\
		& =  \frac{1}{n} \cdot \rho^B(S), \nonumber 
	\end{align}
	where Eq.~\eqref{eq:bpimimm} is due to the definitions of the P-RR set and 
	$\Gamma^B_W(S)$, the set of nodes that are first reached by some node in $S$.
\end{proof}

\begin{proof}[Proof of Theorem~\ref{thm:IMMBPIM} (Sketch)]
Lemma~\ref{lem:bpim} has the same form as Eq.~\eqref{eq:RRsetSpread}, and the boosted preemptive influence spread
	is monotone submodular (Lemma~\ref{lem:properties}), same as the classical influence spread.
Therefore, our {\IMMBPIM} algorithm runs the same as {\IMM} and achieves the same approximation guarantee result.
The only difference is in the generation of P-RR sets, which affects time complexity.
For a P-RR set $R^P$, Let $R$ be the corresponding RR set that contains all nodes visited during the reverse
	simulation process.
Similar to the analysis in~\cite{tang15}, let the width of $R$, $\omega(R)$, be the number of incoming edges pointing to nodes in $R$.
Let $EPT = \E[\omega(R)]$.
In {\IMM}, $EPT$ would be the expected running time of generating one RR set.
But in our algorithm, in particular in {\PRR} algorithm (Algorithm~\ref{alg:pwrr}), we need to use a priority queue for set $Q$
	to support element insertion, deletion, update, and finding min operations, same as the implementation of the Dijkstra
	shortest path algorithm.
Therefore, the expected running time of generating one P-RR set is $O(\E[\omega(R) + |R| \log |R|]) = 
	O(\E[\omega(R) \log n]) = O(EPT \log n)$, leading to an extra $\log n$ factor.
From~\cite{tang15}, we further know that $EPT = m \E[\sigma(\tilde{v})]/n$, where $\tilde{v}$ is a node randomly selected
	with probability proportional to its indegree, and $\sigma(\tilde{v})$ is the classical influence spread of $\tilde{v}$
	in the corresponding IC model.
Let $\theta$ be the number of P-RR sets generated, and let $\OPT = \rho^B(S^*)$ be the optimal solution for the BPIM problem.
Based on the {\IMM} analysis in~\cite{tang15}, we know that the total expected running time is given by 
\begin{align*}
&O(\E[\theta] \cdot \E[\omega(R) + |R| \log |R|]) \\
& = O\left(\frac{(k+\ell)n \log n}{\OPT \cdot \varepsilon^2} \cdot EPT \log n \right) \\
& = O\left( \frac{(k+\ell)(m+n)\log^2 n}{\varepsilon^2} \cdot \frac{\E[\sigma(\tilde{v})]}{\OPT}  \right).
\end{align*}	
Finally, by our assumption that the optimal solution is at least as large as $\E[\sigma(\tilde{v})]$, 
	 the time complexity is as stated in the theorem.
\end{proof}

\begin{proof}[Proof of Lemma~\ref{lem:pim}]
	\begin{align}
		& \E_{W\sim \cW(q,\Delta,p,D)}[\I\{ u = u^s_W\}]  \nonumber \\
		& = \frac{1}{n} \sum_{v\in V}\cdot \E[\I\{ u = u^s_W(v) \}]
		\nonumber \\
		& = \frac{1}{n} \sum_{u\in V}\E[ \I\{v \in \Gamma_W(\{u\})\} ] \nonumber \allowdisplaybreaks \\
		& = \frac{1}{n} \cdot \E[|\Gamma_W(u)| ] \nonumber \\
		& =  \frac{1}{n} \cdot \rho(\{u\}). \nonumber 
	\end{align}
\end{proof}

\begin{proof}[Proof of Theorem~\ref{thm:IMMPIM} (Sketch)]
By Lemma~\ref{lem:properties}, the preemptive influence spread function $\rho$ is additive.
Therefore, we can find top $k$ nodes with the largest preemptive influence spread individually, and together they form 
	the optimal solution for the PIM problem.
To estimate individual node's preemptive influence spread, we utilize Lemma~\ref{lem:pim}, such that the number of times
	a node $u$ is identified as the source node $u^s$ by the {\PRR} algorithm directly corresponds the preemptive
	influence spread $\rho({u})$.
Therefore, we only need to replace the {\NodeSelection} procedure with counting the appearances of each node as sources
	(via variable $est_u$) and selecting the top $k$ of them.
The parameter $\lambda^*(\ell)$ in the original {\IMM} algorithm is also replaced with $\tilde{\lambda}^*$ by
	replacing $(1-1/e)$ with $1$, because we find
	exact solution of top $k$ nodes appearing in most reverse simulations instead of a $1-1/e$ approximation.
With this, following the proof structure of the {\IMM} algorithm, we can show that our {\IMMPIM} algorithm guarantees
	$1-\varepsilon$ approximation.

As for time complexity, the analysis follows the same way as in the proof of Theorem~\ref{thm:IMMBPIM}.
That is, because of {\PRR} algorithm needs to use a priority queue to run a Dijkstra-like algorithm, we need an extra factor
	$\log n$ in the time complexity.
\end{proof}

\begin{table*}[t]
	\centering
	\caption{The performance of influence spread on NetHEPT.} \label{tab:heptspread}
	\resizebox{6.5in}{!}{%
		\begin{tabular}{ccccccc}
			\hline
			\multicolumn{1}{|c|}{NetHEPT}                      & \multicolumn{1}{c|}{Case}               & \multicolumn{1}{c|}{Method}   & \multicolumn{1}{c|}{10}             & \multicolumn{1}{c|}{50}              & \multicolumn{1}{c|}{100}            & \multicolumn{1}{c|}{200}            \\ \hline
			\multicolumn{1}{|c|}{\multirow{15}{*}{PIM}} & \multicolumn{1}{c|}{\multirow{3}{*}{0}} & \multicolumn{1}{c|}{\IMM}      & \multicolumn{1}{c|}{23.35 $\pm$ 0.40}   & \multicolumn{1}{c|}{75.19 $\pm$ 0.64}    & \multicolumn{1}{c|}{130.78 $\pm$ 0.78}  & \multicolumn{1}{c|}{208.24 $\pm$ 0.91}  \\ \cline{3-7} 
			\multicolumn{1}{|c|}{}                      & \multicolumn{1}{c|}{}                   & \multicolumn{1}{c|}{{\ASVRR}}   & \multicolumn{1}{c|}{25.09 $\pm$ 0.41}   & \multicolumn{1}{c|}{73.30 $\pm$ 0.64}    & \multicolumn{1}{c|}{131.96 $\pm$ 0.79}  & \multicolumn{1}{c|}{220.53 $\pm$ 0.96}  \\ \cline{3-7} 
			\multicolumn{1}{|c|}{}                      & \multicolumn{1}{c|}{}                   & \multicolumn{1}{c|}{{\IMMPIM}}  & \multicolumn{1}{c|}{30.32 $\pm$ 0.40 ($+ 20.1\%$)}   & \multicolumn{1}{c|}{108.44 $\pm$ 0.71 ($+ 44.2\%$)}   & \multicolumn{1}{c|}{184.17 $\pm$ 0.88 ($+ 39.6\%$)}  & \multicolumn{1}{c|}{289.16 $\pm$ 1.08 ($+ 31.1\%$)}  \\ \cline{2-7} 
			\multicolumn{1}{|c|}{}                      & \multicolumn{1}{c|}{\multirow{3}{*}{1}} & \multicolumn{1}{c|}{\IMM}      & \multicolumn{1}{c|}{59.26 $\pm$ 0.41}   & \multicolumn{1}{c|}{188.02 $\pm$ 0.69}   & \multicolumn{1}{c|}{323.70 $\pm$ 0.86}  & \multicolumn{1}{c|}{506.95 $\pm$ 1.01}  \\ \cline{3-7} 
			\multicolumn{1}{|c|}{}                      & \multicolumn{1}{c|}{}                   & \multicolumn{1}{c|}{{\ASVRR}}   & \multicolumn{1}{c|}{62.37 $\pm$ 0.41}   & \multicolumn{1}{c|}{185.44 $\pm$ 0.68}   & \multicolumn{1}{c|}{324.81 $\pm$ 0.86}  & \multicolumn{1}{c|}{534.16 $\pm$ 1.04}  \\ \cline{3-7} 
			\multicolumn{1}{|c|}{}                      & \multicolumn{1}{c|}{}                   & \multicolumn{1}{c|}{{\IMMPIM}}  & \multicolumn{1}{c|}{75.18 $\pm$ 0.42 ($+ 20.5\%$)}   & \multicolumn{1}{c|}{268.35 $\pm$ 0.73 ($+ 42.7\%$)}   & \multicolumn{1}{c|}{440.31 $\pm$ 0.91 ($+ 35.6\%$)}  & \multicolumn{1}{c|}{705.06 $\pm$ 1.10 ($+ 32.0\%$)}  \\ \cline{2-7} 
			\multicolumn{1}{|c|}{}                      & \multicolumn{1}{c|}{\multirow{3}{*}{2}} & \multicolumn{1}{c|}{\IMM}      & \multicolumn{1}{c|}{7.71 $\pm$ 0.27}    & \multicolumn{1}{c|}{25.33 $\pm$ 0.43}    & \multicolumn{1}{c|}{45.42 $\pm$ 0.54}   & \multicolumn{1}{c|}{76.45 $\pm$ 0.62}   \\ \cline{3-7} 
			\multicolumn{1}{|c|}{}                      & \multicolumn{1}{c|}{}                   & \multicolumn{1}{c|}{{\ASVRR}}   & \multicolumn{1}{c|}{8.04 $\pm$ 0.28}    & \multicolumn{1}{c|}{25.00 $\pm$ 0.44}    & \multicolumn{1}{c|}{44.81 $\pm$ 0.55}   & \multicolumn{1}{c|}{79.03 $\pm$ 0.69}   \\ \cline{3-7} 
			\multicolumn{1}{|c|}{}                      & \multicolumn{1}{c|}{}                   & \multicolumn{1}{c|}{{\IMMPIM}}  & \multicolumn{1}{c|}{10.66 $\pm$ 0.25 ($+ 55.8\%$)}   & \multicolumn{1}{c|}{39.96 $\pm$ 0.48 ($+ 86.5\%$)}    & \multicolumn{1}{c|}{69.877 $\pm$ 0.61 ($+ 69.1\%$)}  & \multicolumn{1}{c|}{123.14 $\pm$ 0.70 ($+ 56.9\%$)}  \\ \cline{2-7} 
			\multicolumn{1}{|c|}{}                      & \multicolumn{1}{c|}{\multirow{3}{*}{3}} & \multicolumn{1}{c|}{\IMM}      & \multicolumn{1}{c|}{49.93 $\pm$ 0.44}   & \multicolumn{1}{c|}{135.20 $\pm$ 0.67}   & \multicolumn{1}{c|}{234.35 $\pm$ 0.84}  & \multicolumn{1}{c|}{371.37 $\pm$ 0.98}  \\ \cline{3-7} 
			\multicolumn{1}{|c|}{}                      & \multicolumn{1}{c|}{}                   & \multicolumn{1}{c|}{{\ASVRR}}   & \multicolumn{1}{c|}{43.23 $\pm$ 0.41}   & \multicolumn{1}{c|}{139.90 $\pm$ 0.67}   & \multicolumn{1}{c|}{245.02 $\pm$ 0.85}  & \multicolumn{1}{c|}{407.87 $\pm$ 1.04}  \\ \cline{3-7} 
			\multicolumn{1}{|c|}{}                      & \multicolumn{1}{c|}{}                   & \multicolumn{1}{c|}{{\IMMPIM}}  & \multicolumn{1}{c|}{77.81 $\pm$ 0.48 ($+ 32.6\%$)}   & \multicolumn{1}{c|}{260.85 $\pm$ 0.79 ($+ 57.8\%$)}   & \multicolumn{1}{c|}{414.33 $\pm$ 0.97 ($+ 53.8\%$)}  & \multicolumn{1}{c|}{640.14 $\pm$ 1.15 ($+ 55.8\%$)}  \\ \cline{2-7} 
			\multicolumn{1}{|c|}{}                      & \multicolumn{1}{c|}{\multirow{3}{*}{4}} & \multicolumn{1}{c|}{\IMM}      & \multicolumn{1}{c|}{12.24 $\pm$ 0.30}   & \multicolumn{1}{c|}{53.80 $\pm$ 0.56}    & \multicolumn{1}{c|}{92.24 $\pm$ 0.69}   & \multicolumn{1}{c|}{145.59 $\pm$ 0.79}  \\ \cline{3-7} 
			\multicolumn{1}{|c|}{}                      & \multicolumn{1}{c|}{}                   & \multicolumn{1}{c|}{{\ASVRR}}   & \multicolumn{1}{c|}{18.88 $\pm$ 0.38}   & \multicolumn{1}{c|}{49.15 $\pm$ 0.56}    & \multicolumn{1}{c|}{89.53 $\pm$ 0.70}   & \multicolumn{1}{c|}{148.14 $\pm$ 0.84}  \\ \cline{3-7} 
			\multicolumn{1}{|c|}{}                      & \multicolumn{1}{c|}{}                   & \multicolumn{1}{c|}{{\IMMPIM}}  & \multicolumn{1}{c|}{29.82 $\pm$ 0.44 ($+ 57.9\%$)}               & \multicolumn{1}{c|}{97.93 $\pm$ 0.70 ($+ 82.0\%$)}                & \multicolumn{1}{c|}{159.18 $\pm$ 0.82 ($+ 72.6\%$)}               & \multicolumn{1}{c|}{258.13 $\pm$ 0.96 ($+ 74.2\%$)}               \\ 
			\hline 
			\hline 
			\multicolumn{1}{|c|}{\multirow{6}{*}{BIM}}  & \multicolumn{1}{c|}{\multirow{2}{*}{0}} & \multicolumn{1}{c|}{\IMM}      & \multicolumn{1}{c|}{3976.01 $\pm$ 2.01} & \multicolumn{1}{c|}{4295.73 $\pm$ 1.86}  & \multicolumn{1}{c|}{4606.85 $\pm$ 1.77} & \multicolumn{1}{c|}{5094.04 $\pm$ 1.69} \\ \cline{3-7} 
			\multicolumn{1}{|c|}{}                      & \multicolumn{1}{c|}{}                   & \multicolumn{1}{c|}{{\IMMBIM}}  & \multicolumn{1}{c|}{3976.77 $\pm$ 2.01 ($< 1\%$)} & \multicolumn{1}{c|}{4299.85 $\pm$ 1.90 ($< 1\%$)}  & \multicolumn{1}{c|}{4612.03 $\pm$ 1.81 ($< 1\%$)} & \multicolumn{1}{c|}{5097.34 $\pm$ 1.70 ($< 1\%$)} \\ \cline{2-7} 
			\multicolumn{1}{|c|}{}                      & \multicolumn{1}{c|}{\multirow{2}{*}{1}} & \multicolumn{1}{c|}{\IMM}      & \multicolumn{1}{c|}{5526.53 $\pm$ 1.83} & \multicolumn{1}{c|}{5678.88 $\pm$ 1.79}  & \multicolumn{1}{c|}{5840.28 $\pm$ 1.68} & \multicolumn{1}{c|}{6144.55 $\pm$ 1.61} \\ \cline{3-7} 
			\multicolumn{1}{|c|}{}                      & \multicolumn{1}{c|}{}                   & \multicolumn{1}{c|}{{\IMMBIM}}  & \multicolumn{1}{c|}{5547.92 $\pm$ 1.83 ($< 1\%$)} & \multicolumn{1}{c|}{5744.09 $\pm$ 1.78 ($+ 1.1\%$)}  & \multicolumn{1}{c|}{5938.97 $\pm$ 1.72 ($+ 1.7\%$)} & \multicolumn{1}{c|}{6270.64 $\pm$ 1.66 ($+ 2.1\%$)} \\ \cline{2-7} 
			\multicolumn{1}{|c|}{}                      & \multicolumn{1}{c|}{\multirow{2}{*}{3}} & \multicolumn{1}{c|}{\IMM}      & \multicolumn{1}{c|}{4829.76 $\pm$ 1.91} & \multicolumn{1}{c|}{5047.73  $\pm$ 1.82} & \multicolumn{1}{c|}{5272.48 $\pm$ 1.74} & \multicolumn{1}{c|}{5656.25 $\pm$ 1.65} \\ \cline{3-7} 
			\multicolumn{1}{|c|}{}                      & \multicolumn{1}{c|}{}                   & \multicolumn{1}{c|}{{\IMMBIM}}  & \multicolumn{1}{c|}{4841.66 $\pm$ 1.91 ($< 1\%$)} & \multicolumn{1}{c|}{5087.22  $\pm$ 1.85 ($< 1\%$)} & \multicolumn{1}{c|}{5330.86 $\pm$ 1.78 ($+ 1.1\%$)} & \multicolumn{1}{c|}{5739.56 $\pm$ 1.68 ($+ 1.5\%$)} \\ \hline
			\hline
			\multicolumn{1}{|c|}{\multirow{6}{*}{BPIM}} & \multicolumn{1}{c|}{\multirow{2}{*}{0}} & \multicolumn{1}{c|}{\IMM}      & \multicolumn{1}{c|}{188.14 $\pm$ 0.63}  & \multicolumn{1}{c|}{677.81  $\pm$ 0.96}  & \multicolumn{1}{c|}{1128.41 $\pm$ 1.07} & \multicolumn{1}{c|}{1813.26 $\pm$ 1.16} \\ \cline{3-7} 
			\multicolumn{1}{|c|}{}                      & \multicolumn{1}{c|}{}                   & \multicolumn{1}{c|}{{\IMMBPIM}} & \multicolumn{1}{c|}{189.12 $\pm$ 0.64 ($< 1\%$)}  & \multicolumn{1}{c|}{677.67 $\pm$ 0.94 ($< 1\%$)}   & \multicolumn{1}{c|}{1134.77 $\pm$ 1.05 ($< 1\%$)} & \multicolumn{1}{c|}{1819.71 $\pm$ 1.14 ($< 1\%$)} \\ \cline{2-7} 
			\multicolumn{1}{|c|}{}                      & \multicolumn{1}{c|}{\multirow{2}{*}{1}} & \multicolumn{1}{c|}{\IMM}      & \multicolumn{1}{c|}{127.04 $\pm$ 0.43}  & \multicolumn{1}{c|}{508.31 $\pm$ 0.72}   & \multicolumn{1}{c|}{893.64 $\pm$ 0.86}  & \multicolumn{1}{c|}{1509.87 $\pm$ 1.00} \\ \cline{3-7} 
			\multicolumn{1}{|c|}{}                      & \multicolumn{1}{c|}{}                   & \multicolumn{1}{c|}{{\IMMBPIM}} & \multicolumn{1}{c|}{132.34 $\pm$ 0.42 ($+ 4.2\%$)}  & \multicolumn{1}{c|}{517.91 $\pm$ 0.70 ($+ 1.9\%$)}   & \multicolumn{1}{c|}{909.68 $\pm$ 0.84 ($+ 1.8\%$)}  & \multicolumn{1}{c|}{1530.55 $\pm$ 0.96 ($+ 1.4\%$)} \\ \cline{2-7} 
			\multicolumn{1}{|c|}{}                      & \multicolumn{1}{c|}{\multirow{2}{*}{3}} & \multicolumn{1}{c|}{\IMM}      & \multicolumn{1}{c|}{153.42 $\pm$ 0.51}  & \multicolumn{1}{c|}{578.23 $\pm$ 0.80}   & \multicolumn{1}{c|}{995.54 $\pm$ 0.95}  & \multicolumn{1}{c|}{1646.73 $\pm$ 1.08} \\ \cline{3-7} 
			\multicolumn{1}{|c|}{}                      & \multicolumn{1}{c|}{}                   & \multicolumn{1}{c|}{{\IMMBPIM}} & \multicolumn{1}{c|}{155.73 $\pm$ 0.49 ($+ 1.5\%$)}  & \multicolumn{1}{c|}{583.91 $\pm$ 0.79 ($+ 1.0\%$)}   & \multicolumn{1}{c|}{1009.91 $\pm$ 0.94 ($+ 1.4\%$)} & \multicolumn{1}{c|}{1667.51 $\pm$ 1.04 ($+ 1.3\%$)} \\ 
			\hline
		\end{tabular}
	}
\end{table*}

\begin{table*}[t]
	\centering
	\caption{The performance of influence spread on Flixster.}
	\label{tab:flixspread}
	\resizebox{6.5in}{!}{%
		\begin{tabular}{ccccccc}
			\hline
			\multicolumn{1}{|c|}{Flixster}                      & \multicolumn{1}{c|}{Case}               & \multicolumn{1}{c|}{Method}   & \multicolumn{1}{c|}{10}             & \multicolumn{1}{c|}{50}              & \multicolumn{1}{c|}{100}            & \multicolumn{1}{c|}{200}            \\
			\hline
			\multicolumn{1}{|c|}{\multirow{6}{*}{PIM}}  & \multicolumn{1}{c|}{\multirow{3}{*}{0}} & \multicolumn{1}{c|}{IMM}      & \multicolumn{1}{c|}{32.03 $\pm$ 0.39}   & \multicolumn{1}{c|}{110.89 $\pm$ 0.6}   & \multicolumn{1}{c|}{166.58 $\pm$ 0.68}  & \multicolumn{1}{c|}{252.28 $\pm$ 0.78}  \\ \cline{3-7} 
			\multicolumn{1}{|c|}{}                      & \multicolumn{1}{c|}{}                   & \multicolumn{1}{c|}{ASV-RR}   & \multicolumn{1}{c|}{31.59 $\pm$ 0.39}   & \multicolumn{1}{c|}{119.27 $\pm$ 0.64}  & \multicolumn{1}{c|}{170.25 $\pm$ 0.71}  & \multicolumn{1}{c|}{281.37 $\pm$ 0.84}  \\ \cline{3-7} 
			\multicolumn{1}{|c|}{}                      & \multicolumn{1}{c|}{}                   & \multicolumn{1}{c|}{IMM-PIM}  & \multicolumn{1}{c|}{44.1 $\pm$ 0.43 ($+ 37.7\%$)}   & \multicolumn{1}{c|}{142.37 $\pm$ 0.65  ($+ 19.4\%$)} & \multicolumn{1}{c|}{223.52 $\pm$ 0.73 ($+ 31.3\%$)}  & \multicolumn{1}{c|}{339.79 $\pm$ 0.82 ($+ 20.8\%$)}  \\ \cline{2-7} 
			\multicolumn{1}{|c|}{}                      & \multicolumn{1}{c|}{\multirow{3}{*}{3}} & \multicolumn{1}{c|}{IMM}      & \multicolumn{1}{c|}{51.68 $\pm$ 0.33}   & \multicolumn{1}{c|}{165.15 $\pm$ 0.51}  & \multicolumn{1}{c|}{259.12 $\pm$ 0.59}  & \multicolumn{1}{c|}{399.44 $\pm$ 0.67}  \\ \cline{3-7} 
			\multicolumn{1}{|c|}{}                      & \multicolumn{1}{c|}{}                   & \multicolumn{1}{c|}{ASV-RR}   & \multicolumn{1}{c|}{54.5 $\pm$ 0.33}    & \multicolumn{1}{c|}{175.98 $\pm$ 0.53}  & \multicolumn{1}{c|}{256.63 $\pm$ 0.6}   & \multicolumn{1}{c|}{442.88 $\pm$ 0.69}  \\ \cline{3-7} 
			\multicolumn{1}{|c|}{}                      & \multicolumn{1}{c|}{}                   & \multicolumn{1}{c|}{IMM-PIM}  & \multicolumn{1}{c|}{105.39 $\pm$ 0.34  ($+ 93.4\%$)}  & \multicolumn{1}{c|}{311.99 $\pm$ 0.54 ($+ 77.3\%$)}  & \multicolumn{1}{c|}{475.21 $\pm$ 0.64 ($+ 83.4\%$)}  & \multicolumn{1}{c|}{702.76 $\pm$ 0.71 ($+ 58.7\%$)} \\ \hline
			\hline
			
			\multicolumn{1}{|c|}{\multirow{4}{*}{BIM}}  & \multicolumn{1}{c|}{\multirow{2}{*}{0}} & \multicolumn{1}{c|}{IMM}      & \multicolumn{1}{c|}{4693.78 $\pm$ 1.13} & \multicolumn{1}{c|}{5041.03 $\pm$ 1.02} & \multicolumn{1}{c|}{5315.41 $\pm$ 0.96} & \multicolumn{1}{c|}{5712.07 $\pm$ 0.92} \\ \cline{3-7} 
			\multicolumn{1}{|c|}{}                      & \multicolumn{1}{c|}{}                   & \multicolumn{1}{c|}{IMM-BIM}  & \multicolumn{1}{c|}{4694.19 $\pm$ 1.14 ($< 1\%$)} & \multicolumn{1}{c|}{5046.5 $\pm$ 1.04 ($< 1\%$)}  & \multicolumn{1}{c|}{5327.35 $\pm$ 0.98 ($< 1\%$)} & \multicolumn{1}{c|}{5718.48 $\pm$ 0.93 ($< 1\%$)} \\ \cline{2-7} 
			\multicolumn{1}{|c|}{}                      & \multicolumn{1}{c|}{\multirow{2}{*}{3}} & \multicolumn{1}{c|}{IMM}      & \multicolumn{1}{c|}{6892.51 $\pm$ 1.01} & \multicolumn{1}{c|}{7050.49 $\pm$ 0.96} & \multicolumn{1}{c|}{7200.86 $\pm$ 0.92} & \multicolumn{1}{c|}{7471.4 $\pm$ 0.9}   \\ \cline{3-7} 
			\multicolumn{1}{|c|}{}                      & \multicolumn{1}{c|}{}                   & \multicolumn{1}{c|}{IMM-BIM}  & \multicolumn{1}{c|}{6898.82 $\pm$ 1.01 ($< 1\%$)} & \multicolumn{1}{c|}{7081.8 $\pm$ 0.98 ($< 1\%$)}  & \multicolumn{1}{c|}{7258.11 $\pm$ 0.96 ($< 1\%$)} & \multicolumn{1}{c|}{7543.74 $\pm$ 0.91 ($+ 1.0\%$)} \\ \hline
			\hline
			\multicolumn{1}{|c|}{\multirow{4}{*}{BPIM}} & \multicolumn{1}{c|}{\multirow{2}{*}{0}} & \multicolumn{1}{c|}{IMM}      & \multicolumn{1}{c|}{294.41 $\pm$ 0.49}  & \multicolumn{1}{c|}{877.56 $\pm$ 0.7}   & \multicolumn{1}{c|}{1293.56 $\pm$ 0.74} & \multicolumn{1}{c|}{1860.55 $\pm$ 0.76} \\ \cline{3-7} 
			\multicolumn{1}{|c|}{}                      & \multicolumn{1}{c|}{}                   & \multicolumn{1}{c|}{IMM-BPIM} & \multicolumn{1}{c|}{299.74 $\pm$ 0.48 ($+ 1.8\%$)}  & \multicolumn{1}{c|}{887.56 $\pm$ 0.68 ($+ 1.1\%$)}  & \multicolumn{1}{c|}{1307.58 $\pm$ 0.71 ($+ 1.1\%$)} & \multicolumn{1}{c|}{1875.77 $\pm$ 0.72 ($< 1\%$)} \\ \cline{2-7} 
			\multicolumn{1}{|c|}{}                      & \multicolumn{1}{c|}{\multirow{2}{*}{3}} & \multicolumn{1}{c|}{IMM}      & \multicolumn{1}{c|}{194.11 $\pm$ 0.34}  & \multicolumn{1}{c|}{599.17 $\pm$ 0.54}  & \multicolumn{1}{c|}{925.57 $\pm$ 0.61}  & \multicolumn{1}{c|}{1414.42 $\pm$ 0.65} \\ \cline{3-7} 
			\multicolumn{1}{|c|}{}                      & \multicolumn{1}{c|}{}                   & \multicolumn{1}{c|}{IMM-BPIM} & \multicolumn{1}{c|}{201.02 $\pm$ 0.34 ($+ 3.6\%$)}  & \multicolumn{1}{c|}{625.58 $\pm$ 0.53 ($+ 4.4\%$)}  & \multicolumn{1}{c|}{972.24 $\pm$ 0.58 ($+ 5.0\%$)}  & \multicolumn{1}{c|}{1473.65 $\pm$ 0.64 ($+ 4.2\%$)} \\ \hline
		\end{tabular}
	}
\end{table*}

\begin{table*}[t!]
	\centering
	\caption{Running time results on datasets (in seconds).} \label{tab:time2}
	\resizebox{4.5in}{!}{
		\begin{tabular}{|l|c|c|c|c|c|}
			\hline
			Data & {\sf \IMMBIM}     & {\sf \IMMPIM}    & {\sf \IMMBPIM}   & {\IMM}     & {\ASVRR} \\ \hline
			NetHEPT  & 0.7534  & 195.75 & 48.123 & 1.9712 & 74.175 \\ \hline
			Flixster & 1.3516  & 955.45 & 218.51 & 5.1072  & 235.34  \\ \hline
		\end{tabular}
	}
\end{table*}

\begin{table*}[t]
	\centering
	\caption{Running time results on $\varepsilon$ value (in seconds).} \label{tab:time3}
	\resizebox{6.5in}{!}{
		\begin{tabular}{|l|c|c|c|c|c|}
			\hline
			NetHEPT ($\varepsilon$)  & 0.1              & 0.2             & 0.3             & 0.4             & 0.5             \\ \hline
			\IMMBIM  & 5098.01 (0.73s)  & 5075.54 (0.28s) & 5063.49 (0.21s) & 5037.64 (0.18s) & 5004.03 (0.13s) \\ \hline
			\IMMPIM  & 289.16 (492.49s) & 290.67 (136.1s)    & 287.61 (67.5s)    & 288.22 (41.1s)    & 287.33 (29.2s) \\ \hline
			\IMMBPIM & 1819.19 (53.6s)  & 1811.79 (14.3s) & 1792.67 (7.1s)  & 1781.42 (4.7s)  & 1768.05 (3.4s)  \\ \hline
		\end{tabular}
	}
\end{table*}

\begin{table}[t!]
	\centering
	\caption{Running time results on self-activation probability distribution (in seconds)} \label{table:time4}
	\resizebox{3.5in}{!}{
		\begin{tabular}{|c|c|c|c|c|c|}
			\hline
			{}$\alpha_u$ & {[}0-0.2{]} & {[}0.2-0.4{]} & {[}0.4-0.6{]} & {[}0.6-0.8{]} & {[}0.8-1.0{]} \\ \hline
			\IMMBIM                     & 0.839       & 0.456         & 0.328         & 0.259         & 0.226         \\ \hline
			\IMMBPIM                    & 53.92       & 63.47         & 73.22         & 117.03        & 126.82        \\ \hline
			\IMMPIM                     & 492.49      & 378.04        & 347.67        & 314.58        & 296.72        \\ \hline
		\end{tabular}
	}
\end{table}

\section{Additional Experimental Results}

\noindent
{\bf Self-activation parameters and test cases. \ \ }
In practice, self-activation delays can be estimated from the users' access patterns to online
social networks, and self-activation probabilities can be estimated from the
fraction of times users' participating in information cascades not due to the influence
from the neighbors or external selections as seed users.
Unfortunately for the datasets we use, these information are not available.
Instead, we use synthetic settings, focusing on whether the knowledge of the self-activation
behaviors would benefit our algorithm design.
For self-activation probability $q(u)$ of node $u$, we 
first randomly select a value $\beta_u$ from $[0,c]$ as
a node $u$'s base value, then we further test five cases:
(0) uniform: $q(u) = \alpha_u$; 
(1) positively correlated: $q(u)$ is positively correlated with $u$'s out-degree
$d^+(u)$, in particular $q(u) = \min\{\beta_u \cdot d^+(u), 1\}$;
(2) negatively correlated: $q(u) = \beta_u / d^+(u)$;
(3) random mixing of cases 0 and 1: randomly pick half of the nodes with 
$q(u) = \alpha_u$ and the other half with $q(u) = \min\{\beta_u \cdot d^+(u), 1\}$;
(4) random mixing of cases 0 and 2: randomly pick half of the nodes with 
$q(u) = \alpha_u$ and the other half with $q(u) = \beta_u / d^+(u)$.
These five cases aim at scenarios where all users are equally likely to react to
a campaign (case 0), high-degree nodes (usually more influential) are more likely
or less likely to react to the campaigns, and mixture of uniform behavior and a 
correlation (or reverse correlation) behavior.
We set $c = 2$ for BPIM and PIM tests,
because otherwise the preemptive influence spread for PIM is too small.
For self-activation delays, we use exponential distribution with rate $1$ for all nodes.
We already vary the self-activation behaviors through the self-activation probabilities,
and thus we simply keep the self-activation delay distributions uniform.
We also use the same exponential distribution for propagation delay distributions.

\subsection{Effectiveness Analysis}

We provide the full test case results for NetHEPT in Table~\ref{tab:heptspread}.
Each influence spread result is an average over 10000 simulations runs for NetHEPT and
Flixter
The number following the $\pm$ sign is the $95\%$ confidence radius, and
the percentage in parenthesis is the improvement of our algorithm over the best
baseline result given above it.
The results on PIM clearly show that our {\IMMPIM} algorithm has improvement in
the preemptive influence spread achieved comparing against baselines {\IMM} and
{\ASVRR} --- the improvements is from $20\%$ to $98.1\%$.
This shows that, when identifying top influencers in the presence of self activation behavior, it is important to incorporate the knowledge about self activation into the algorithm design

For BIM and BPIM, the improvement is less significant. In particular, the uniform case
does not show significant improvement, while for positively correlated cases, 
the improvement could be $1\%$ to $4\%$.
This suggests that, due to the boost of seed nodes, the knowledge of natural self-activation
behaviors become less important, and some times it is ok to use the self-activation
oblivious {\IMM} algorithm, but if a few percentage of improvement is still important,
the extra knowledge on self-activation together with our algorithms are still beneficial.
We ignore cases 2 and 4, and in general the improvement is less significant.

The results on Flixster are similar, we list the cases 0 and 3 results in Tables~\ref{tab:flixspread}
which show significant improvements over {\IMM} and {\ASVRR} for PIM
and outperform {\IMM} for BIM and BPIM.

\subsection{Running Time Analysis}

Table~\ref{tab:time2} reports the running time of all algorithms on both datasets,
by using the default setting with the seed set size $k = 200$ in test case 3.
We can clearly see the order of running time is $\IMMBIM < \IMM < \IMMBPIM < \ASVRR < \IMMPIM$
(we ignore the prefix {\IMM} in our algorithm to fit the table width).
This is inline with our theoretical analysis, which shows that the running time is
inversely proportional to the optimal value of each problem.
For example, the optimal solution of BIM is larger than that of the classical 
influence maximization because BIM has self-activated nodes contributing extra
influence spread, and PIM has the smallest optimal value because 
the self-activation probabilities are small in general and the optimal set has
to compete with other self-activated nodes on preemptive influence spread.
{\IMMBPIM} and {\IMMPIM} are further slower due to the Dijkstra-like reverse simulation, which
takes more time than the simple breadth-first-search simulation.
But even for the slowest {\IMMPIM} algorithm, on the Flixster dataset with more than hundred thousands nodes and edges, it could complete in less than 16 minutes on our laptop test machine.
Besides the dataset, running time is also related to the $\varepsilon$.

As shown in Table \ref{tab:time3}, we test $\varepsilon$ values from 0.1 to 0.5 on NetHEPT for the proposed algorithms.
We can see influence difference is less than 3\% when $\varepsilon$ increases from 0.1 to 0.5 for {\IMMBIM} and {\IMMBPIM},
and there is almost no influence difference for {\IMMPIM}, which indicates it could be accelerated by increasing $\varepsilon$.

Besides the dataset and $\varepsilon$, the distribution of the self-activation probability $q(u)$ would infect the running time of the proposed algorithm.
We randomly select value from different uniform self-activation probability distributions $[0,0.2]$, $[0.1,0.3]$, $[0.2,0.4]$, $[0.3,0.5]$, $[0.4,0.6]$ for the proposed algorithms.
From Table \ref{table:time4}, we can clearly see that for both {\IMMBIM} and {\IMMPIM}, they spend less running time with larger self-activation probability distribution.
On the contrary, {\IMMBPIM} runs faster with a smaller distribution of self-activation probabilities.
In summary, the running times of the proposed algorithms are sensitive to the self-activation probability distribution of the nodes.